\RequirePackage{fix-cm}
\documentclass[smallextended]{svjour3}       % onecolumn (second format)
\smartqed  % flush right qed marks, e.g. at end of proof
\usepackage{graphicx}
\journalname{Designs, Codes and Cryptography}
\usepackage{amssymb}
\usepackage{amsfonts}
\usepackage{color,graphicx}

\def\newpic#1{\def\emline##1##2##3##4##5##6{\put(##1,##2){\special{em:point #1##3}}\put(##4,##5){\special{em:point #1##6}}\special{em:line #1##3,#1##6}}}
\newpic{}
\def\emline#1#2#3#4#5#6{\put(#1,#2){\special{em:moveto}}\put(#4,#5){\special{em:lineto}}}
\def\newpic#1{}

\begin{document}

\title{A Generalization of Lee Codes%\thanks{Grants or other notes
}


\author{C. Araujo \and
I. Dejter \and P. Horak}

%First Author         \and
%        Second Author %etc.}


\institute{C. Araujo \at
              University of Puerto Rico, Rio Piedras, PR 00936-8377 \\
              Tel.: +787-764-0000\\
              Fax: +787-281-0653\\
              \email{carlosjulio56@gmail.com}           %  \\
           \and
           I. Dejter \at
              University of Puerto Rico, Rio Piedras, PR 00936-8377 \\
              Tel.: +787-764-0000\\
              Fax: +787-281-0653\\
              \email{italo.dejter@gmail.com.com}           %  \\
           \and
           P. Horak \at
           University of Washington, Tacoma, WA 98402 \\
           Tel.:+253-692-4558\\
           \email{horak@uw.edu}
}

\date{Received: date / Accepted: date}

\maketitle

\begin{abstract}
Motivated by a problem in computer architecture we introduce a notion of the
perfect distance-dominating set, PDDS, in a graph. PDDS\thinspace s
constitute a generalization of perfect Lee codes, diameter perfect codes, as
well as other codes and dominating sets. In this paper we initiate a
systematic study of PDDS\thinspace s. PDDS\thinspace s related to the
application will be constructed and the non-existence of some PDDS\thinspace
s will be shown. In addition, an extension of the long-standing Golomb-Welch
conjecture, in terms of PDDS, will be stated. We note that all constructed
PDDS\thinspace s are lattice-like which is a very important feature from the
practical point of view as in this case decoding algorithms tend to be much
simpler.\bigskip

\textbf{This paper is dedicated to the memory of Lucia Gionfriddo}.

\end{abstract}




%First Author         \and
%        Second Author %etc.}


\institute{C. Araujo \at
              University of Puerto Rico, Rio Piedras, PR 00936-8377 \\
              Tel.: +787-764-0000\\
              Fax: +787-281-0653\\
              \email{carlosjulio56@gmail.com}                      \and
           I. Dejter \at
              University of Puerto Rico, Rio Piedras, PR 00936-8377 \\
              Tel.: +787-764-0000\\
              Fax: +787-281-0653\\
              \email{italo.dejter@gmail.com.com}                      \and
           P. Horak \at
           University of Washington, Tacoma, WA 98402 \\
           Tel.:+253-692-4558\\
           \email{horak@uw.edu}
}


\section{Introduction}

\label{intro} \noindent We introduce a generalization of perfect Lee codes
and other dominating notions, motivated by the following problem in computer
architecture, see e.g. \cite{BL}. Processing elements in a supercomputer
communicate through a network that has the topology of the Cartesian product
of cycles. It is desirable to place the Input/Output devices into the
network in such a way that the communication of all elements in the network
is optimized; each element of the network should be at distance at most $t$
from at least one I/O device, ideally from exactly one I/O device. It is not
difficult to see that perfect error correcting Lee codes, if any, provide
the optimal placement.

\bigskip

\noindent Unfortunately, the perfect $t$-error correcting Lee codes of block
length $n$ over $\mathbb{Z}$, and over $\mathbb{Z}_{q},q\geq 2n+1,$ shortly $%
\mathrm{PLC}(n,t)$ and $\mathrm{PLC}(n,t,q)$codes, respectively, have been
constructed only for $n=1,2,$ and any $t,$ and for $n\geq 3$ and $t=1$.
Moreover, as suggested by the well-known and long-standing conjecture of
Golomb and Welch \cite{GW}, $\mathrm{PLC}(n,t)$ codes and $\mathrm{PLC}%
(n,t,q),q\geq 2n+1,$ codes do not exist in other cases. To remedy this
obstacle, perfect Lee codes have been generalized in several ways, see e.g.
\cite{BB}, where the quasi-perfect Lee codes have been introduced. A
weakness of the quasi-perfect Lee codes is that some words cannot be decoded
in a unique way, and so far the quasi-perfect Lee codes have been found only
for $n=2.$\bigskip

\noindent In order to offer a new approach to the placement problem we will
introduce yet another generalization of Lee codes. Instead of defining it
only for the Cartesian product of cycles and the Cartesian product of
two-way infinite paths, denoted by $\Lambda _{n}$ ($=$ infinite graph whose
vertex set is $\mathbb{Z}^{n}$ with two vertices being adjacent if their
Euclidean distance is $1)$, we introduce the new concept for an arbitrary
graph. However, having in mind the application we will mainly focus on the
Cartesian product of cycles and $\Lambda _{n}$. As usual, $[S]$ stands for
the subgraph induced by $S$, and the distance $d(v,C)$ of a vertex $v\in V$
to $C$ is given by $d(v,C)=\min \{d(v,w);w\in C\}$\textrm{.}\

\begin{definition}
Let $t\geq 1$ and $\Gamma =(V,E)$ be a graph. A set $S\subset V$ will be
said to be a $t$-perfect distance-dominating set in $\Gamma $, a $t$-$%
\mathrm{PDDS}$ in $\Gamma ,$ if, for each $v\in V$, there is a unique
component $C_{v}$ of $[S]$, so that for the distance $d(v,C_{v})$ from $v$
to $C_{v}$ it is $d(v,C_{v})\leq t$, and there is in $C_{v}$ a unique vertex
$w$ with $d(v,w)=d(v,C_{v})$.
\end{definition}

\noindent The first condition guaranties that to each element $v$ of the
network there is at least one I/O device at the distance at most $t$ from $%
v, $ while the second condition, that in $C_{v}$ there is a unique vertex $w$
with $d(v,w)=d(v,C_{v}),$ guarantees that to each element $v$ in the
communication network, there is a uniquely determined I/O device with which $%
v$ will communicate.\bigskip

\noindent Now we describe how the new domination concept of $\mathrm{PDDS}$
relates to other coding theory and graph domination notions. First of all we
note that

\noindent $\mathrm{PLC}(n,t,q)$ codes and $\mathrm{PLC}(n,t)$ codes are $t$-$%
\mathrm{PDDS}$\thinspace s in the Cartesian product of cycles and in $%
\Lambda _{n}$, respectively, with all components of $t$-$\mathrm{PDDS}$
being isolated vertices. A notion of a diameter perfect code has been
introduced in \cite{A1}. For $d$ odd, the diameter-$d$ perfect Lee code in $%
\Lambda _{n}$ coincides with the perfect $\frac{d-1}{2}$-error correcting
Lee code. It follows from \cite{A,E} that, for $d$ even, diameter-$d$
perfect Lee code in $\Lambda _{n}$ exists if and only if there is a $\frac{%
d-2}{2}$-$\mathrm{PDDS}$ in $\Lambda _{n}$ whose each component consists of
two adjacent vertices. In \cite{biggs} Biggs extended the concept of the
perfect code from a metric space to a graph. A perfect $t$-code in a graph $%
\Gamma =(V,E)$ is a set $C\subset V$ such that $t$-neighborhoods $%
N_{t}(c)=\{u\in V;d(c,u)\leq t\}$ with $c\in C$ form a partition of $V$.
Clearly, a $t$-perfect code $C$ in $\Gamma $ is a $t$-$\mathrm{PDDS}$ in $%
\Gamma $ with all vertices in $C$ being isolated. Further, Weichsel \cite{W}
defined a notion of the perfect dominating set, or PDS. In our terminology a
PDS is a 1-$\mathrm{PDDS}$. PDS\thinspace s were studied in the hypercube
graphs \cite{W,DW,DP}, in the star graphs \cite{DS}, in $\Lambda _{2},$ and
in toroidal grids \cite{DD,Dsolo}. In addition, Klostermeyer and Goldwasser
\cite{KG} defined the total perfect code in a graph to be a subset of its
vertex set with the property that each vertex is adjacent to exactly one
vertex in the subset. The NP-completeness of finding a 1-perfect code of $%
\Gamma $ and that of finding a minimal perfect dominating set in a planar
graph were established in \cite{BBS,K}, and in \cite{FH},
respectively.\bigskip

\noindent Now we prove a statement related to the structure of $\mathrm{PDDS}
$\thinspace s in $\Lambda _{n}.$ It turns out that the choice of components
of a $t$-$\mathrm{PDDS}$ in $\Lambda _{n}$ is quite limited. To facilitate
our discussion we introduce some notation. If no ambiguity is possible, $n$%
-tuples representing elements of $\mathbb{Z}^{n}$ will be written without
external parentheses or commas. $O$ will stand for the element $00\ldots 0$
and $e_{1}=10\ldots 0$, $e_{2}=010\ldots 0$, $\ldots $, $e_{n}=00\ldots 1$.

\begin{theorem}
If $S$ is a $t$-$\mathrm{PDDS}$ in $\Lambda _{n}$ then each component of $S$
is the Cartesian product of $($possibly infinite$)$ paths.
\end{theorem}

\begin{proof}
Let $S_{0}$ be a component of $S$ in $\Lambda _{n}$. Assume that $S_{0}$ is
not a product of paths. Then wlog we may assume that $O,e_{1}+e_{2}\in S_{0}$%
, and $e_{1}\notin V(S_{0})$. Now, $%
d(e_{1},S_{0})=d(e_{1},O)=d(e_{1},e_{1}+e_{2})=1$. That is, the vertex $v$
is at the minimum distance $1$ from two different vertices of $S$, a
contradiction.
\end{proof}

\noindent A similar result, in the case when PDS of the $n$-dimensional cube
were considered, has been proved in \cite{W}.\bigskip

\noindent With respect to the application mentioned above we will confine
ourselves to the most interesting case of $t$-$\mathrm{PDDS}$\thinspace s in
$\Lambda _{n}$ whose components are all isomorphic to a fixed finite graph $%
H $, denoted for short by $t$-$\mathrm{PDDS}[H]$. It would be very useful to
characterize all finite graphs $H$ for which there is a $t$-$\mathrm{PDDS}%
[H] $. This would show the strength but also limitations of the new concept
for practical purposes. So far we are able to do it only for $\Lambda _{2}.$%
\bigskip

\begin{remark}
\label{RR}We point out that if $R$ is a $t$-$\mathrm{PDDS}[H],$ $H=(V,E),$ then $R$
can be seen as a tiling of $\mathbb{Z}^{n}$ by the graph $H^{\ast }=(V^{\ast
},E^{\ast })$ where $H^{\ast }$ is the indunced subgraph of $\Lambda _{n}$
on the set $V^{\ast },$where  $v\in V^{\ast }$ if and only if $d(v,V)\leq t.$
\end{remark}

\noindent As usual $P_{k}$ will stand for the path on $k$ vertices. Hence, $%
P_{1}$ is an isolated vertex. Further, the cartesian product of graphs $G$
and $H$ is denoted by $G\square H.$ At the moment we do not have enough
evidence to conjecture when a $t$-$\mathrm{PDDS}[H]$ exists in a general
case. However, if $H$ is a product of at most two paths then we strongly
believe that:

\begin{conjecture}
\label{N}Let $H$ be a finite path or a Cartesian product of two finite
paths. Then a $t$-$\mathrm{PDDS}[H]$ in $\Lambda _{n}$ exists if and only if
either $\mathrm{\mathbf{(i)}}$ $t=1$, $n\geq 2,$ and $\ H=P_{k},k\geq 1;$ or
$\mathrm{\mathbf{(ii)}}$ $t\geq 1$, $n=2,$ and $H=P_{k},k\geq 1;$ or $%
\mathrm{\mathbf{(iii)}}$ $t\geq 1$, $n=2,$ and $H=P_{2}\square P_{k}$, $%
k\geq 2$; or $\mathrm{\mathbf{(iv)}}$ $t=1$, $n=3r+2$, $r\geq 0,$ and $%
H=P_{2}\square P_{2};$ or $\mathrm{\mathbf{(v)}}$ $t=2,n=3,$ and $H=P_{2}.$
\end{conjecture}

\noindent We note that $\mathrm{\mathbf{(i)}}$ and $\mathrm{\mathbf{(ii)}}$
extend Golomb-Welch conjecture as well as a conjecture raised in \cite{E} by
Etzion. For $k=1,$ the existence of a $t$-$\mathrm{PDDS}[P_{k}]$ in \textrm{%
(i)} and \textrm{(ii)} was shown by several authors in terms of $\mathrm{PLC}
$ codes, see e.g. Golomb and Welch \cite{GW}, and, for $k=2,$ by Etzion \cite%
{E} in terms of diameter perfect Lee codes. The existence of a $2-\mathrm{%
PDDS}[P_{2}]$ in $\Lambda _{3}$ follows from a Minkowski's tiling \cite{Min}%
.\bigskip\

\noindent The next theorem constitutes one of the main results of the paper.

\begin{theorem}
\label{T}A $t$-$\mathrm{PDDS}[H]$ exists for all graphs $H$ described in
\textrm{Conjecture \ref{N}}.
\end{theorem}

\noindent The following theorem provides additional supporting evidence for
Conjecture \ref{N}.
%We note that the proof of part \textrm{(ii) }$\ $is given in a separate paper \cite{AD0}:

\begin{theorem}
\label{T1}
%$\mathrm{\mathbf{(i)}}$
If $3\leq s\leq r$ then there is no $t$-$ \mathrm{PDDS}[P_{s}\square P_{r}]$\ in $\Lambda _{2}$
for $t\geq 1$.
%\mathrm{\mathbf{(ii)}}$ There is no $1$-$\mathrm{PDDS}[P_{2}\square P_{2}]$ in $\Lambda _{3}.$
\end{theorem}

%\noindent Combining the two above theorems settles in the affirmative
%Conjecture \ref{N} for $n=2.$ Of course, in this case, $H$ cannot be a
%product of more than two paths.

\begin{corollary}
A $t$-$\mathrm{PDDS}[H]$ in $\Lambda _{2}$ exists if and only if either $%
t\geq 1$, and $H=P_{k},k\geq 1,$ or $t\geq 1$, and $H=P_{2}\square
P_{k},k\geq 2.$
\end{corollary}

\noindent To show that a $t$-$\mathrm{PDDS}[H]$ exists also in the case when
$H$ is the Cartesian product of at least three paths we offer the following
theorem:

\begin{theorem}
\label{T2}There is a $1$-$\mathrm{PDDS}[Q_{3}]$ in $\Lambda _{3},$ where $%
Q_{3}=P_{2}\square P_{2}\square P_{2}$ is the $3$-dimensional hypercube.
\end{theorem}

\noindent Recently we learnt that Buzaglo and Etzion proved that a $1$-%
\textrm{PDDS}$[Q_{n}]$ exists if and only if $n=2^{k}-1,$ or $n=3^{k}-1,$
c.f. \cite{BE}. They proved the statement in terms of tilings by crosses;
see Remark \ref{RR}.\bigskip

\noindent All $t$-$\mathrm{PDDS}$\thinspace s constructed in this paper are
lattice-like, which is a very important feature from the practical point of
view as in this case decoding algorithms tend to be much simpler. As the
notion of lattice-like \textrm{$PDDS$} is a key one we provide a formal
definition. Let $H=(V,E)$ be a subgraph of $\Lambda _{n},$ and let $z\in
\mathbb{Z}^{n}$. Then $H+z$ denotes the graph $H^{\prime }=(V^{\prime
},E^{\prime })$, where $V^{\prime }=V+z=\{w$; there exists $v\in V,w=v+z\},$
and $uv\in E$ if and only if $(u+z)(v+z)\in E^{\prime }$. Let $R$ be a $t$-%
\textrm{$PDDS$}$[H]$ and $D\simeq H$ be a component of $R.$ Then $R$ will be
called \textit{lattice-like} if there exists a lattice $L$ such that $%
D^{\prime }$ is a component of $R$ if and only if there is $z\in L$ so that $%
D^{\prime }=D+z$. We recall, see Remark \ref{RR}, that a $t$-\textrm{PDDS}$%
[H]$ can be seen as a tiling. Thus a notion of a lattice-like tiling will be
understood in the same way as a lattice-like \textrm{PDDS}.\bigskip

\noindent All desired $t$-$\mathrm{PDDS}$ in $\Lambda _{n}$ will be
constructed by the same algebraic method. A $\mathrm{PDDS}$ constructed this
way is lattice-like, which in turn implies that such a $\mathrm{PDDS}$ is
periodic as well. That is, a suitable restriction of this $\mathrm{PDDS}$
constitutes a $\mathrm{PDDS}$ in the Cartesian product of cycles. This is
the case of main interest because of the placement problem discussed above.
We recall that a set $S\subset \mathbb{Z}^{n}$ is \textit{periodic} if there
are integers $p_{1},\ldots ,p_{n}$ such that $v\in S$ implies $v\pm
p_{i}e_{i}\in S$ for all $i=1,\ldots ,n,$ where $e_{i}$ is the unit vector
in the direction of the $i$-axis. We recall that each lattice-like $t$-$%
\mathrm{PDDS}$ is periodic, but the converse is not true in
general.\bigskip

\noindent Now we describe a construction of a partition (tiling) of $\Lambda
_{n}$. As far as we know Stein in \cite{SStein} was the first one to use a
group homomorphism to construct a lattice-like tiling; he did it in the case
of a tiling by different types of crosses. Several variations of Stein's
construction can be found throughout the literature, see e.g. \cite%
{SStein,Molnar,Sz,S,HS,Costa,Sw,HB}. For the reader's convenience we
provide a detailed description of this generalization. Let
$(\mathbb{Z}^{n},+)$ be the (com\-po\-nent-wise)
additive group on $\mathbb{Z}^{n}$. Consider a lattice $L$ in $(\mathbb{Z}%
^{n},+)$, i.e. a subgroup of $(\mathbb{Z}^{n},+)$, generated by elements $%
u_{1},\ldots ,u_{n}\in \mathbb{Z}^{n}$; hence $L=\{\alpha _{1}u_{1}+\ldots
+\alpha _{n}u_{n};\alpha _{i}\in \mathbb{Z},i=1,\ldots ,n\}$. We denote by $%
F $ the factor group $(\mathbb{Z}^{n},+)/L$. Furthermore, let a set $T$ of
vertices in $\mathbb{Z}^{n}$ contain one element from each coset of $(%
\mathbb{Z}^{n},+)/L$. Then, $\mathcal{T}=\{T+u\mathbf{;}$ $u\mathbf{\in }L%
\mathbf{\}}$ constitutes a partition of $\mathbb{Z}^{n}$ into parts of size $%
|F|$ and, for each $u\in L$, we have that $[T+u],$ the subgraph of $\Lambda
_{n}$ induced by $T+u$, is isomorphic to $[T]$. Clearly, for a given lattice
$L,$ we can partition the vertex set of $\Lambda _{n}$ into parts such that
the corresponding induced subgraphs have different shapes depending on the
choice of $T$.\bigskip

\noindent Example. Set $L=\{\alpha _{1}\mathbf{(}13,0)+\alpha
_{2}(3,2);\alpha _{i}\in \mathbb{Z},i=1,2\}$. Then, $(\mathbb{Z}%
^{2},+)/L=Z_{13}$. There are many options how to choose the graph $[T],$%
e.g., $[T]$ might be a path of length $12$, or a Lee sphere of radius 2, see
the figure below where the both options are depicted in bold font. The
numbers at the vertices of $\Lambda _{2}$ are elements of $Z_{13}=(\mathbb{Z}%
^{2},+)/L.$

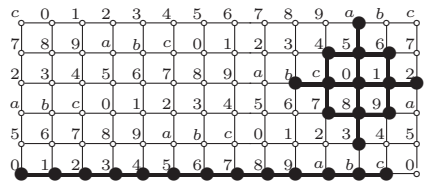
\begin{figure}[tph]
\unitlength=0.40mm \special{em:linewidth 0.4pt} \linethickness{0.4pt}
\begin{picture}(217.00,55.00)
\put(215.00,2.00){\circle{2.00}}
\put(85.00,12.00){\circle{2.00}}
\put(95.00,12.00){\circle{2.00}}
\put(105.00,12.00){\circle{2.00}}
\put(115.00,12.00){\circle{2.00}}
\put(125.00,12.00){\circle{2.00}}
\put(135.00,12.00){\circle{2.00}}
\put(145.00,12.00){\circle{2.00}}
\put(155.00,12.00){\circle{2.00}}
\put(165.00,12.00){\circle{2.00}}
\put(175.00,12.00){\circle{2.00}}
\put(185.00,12.00){\circle{2.00}}
\put(205.00,12.00){\circle{2.00}}
\put(215.00,12.00){\circle{2.00}}
\put(85.00,22.00){\circle{2.00}}
\put(95.00,22.00){\circle{2.00}}
\put(105.00,22.00){\circle{2.00}}
\put(115.00,22.00){\circle{2.00}}
\put(125.00,22.00){\circle{2.00}}
\put(135.00,22.00){\circle{2.00}}
\put(155.00,22.00){\circle{2.00}}
\put(165.00,22.00){\circle{2.00}}
\put(175.00,22.00){\circle{2.00}}
\put(215.00,22.00){\circle{2.00}}
\put(85.00,32.00){\circle{2.00}}
\put(95.00,32.00){\circle{2.00}}
\put(105.00,32.00){\circle{2.00}}
\put(115.00,32.00){\circle{2.00}}
\put(125.00,32.00){\circle{2.00}}
\put(165.00,32.00){\circle{2.00}}
\put(85.00,42.00){\circle{2.00}}
\put(95.00,42.00){\circle{2.00}}
\put(105.00,42.00){\circle{2.00}}
\put(115.00,42.00){\circle{2.00}}
\put(175.00,42.00){\circle{2.00}}
\put(215.00,42.00){\circle{2.00}}
\put(85.00,52.00){\circle{2.00}}
\put(95.00,52.00){\circle{2.00}}
\put(105.00,52.00){\circle{2.00}}
\put(115.00,52.00){\circle{2.00}}
\put(125.00,52.00){\circle{2.00}}
\put(165.00,52.00){\circle{2.00}}
\put(175.00,52.00){\circle{2.00}}
\put(185.00,52.00){\circle{2.00}}
\put(205.00,52.00){\circle{2.00}}
\put(215.00,52.00){\circle{2.00}}
\put(83.00,5.00){\makebox(0,0)[cc]{$_0$}}
\put(93.00,5.00){\makebox(0,0)[cc]{$_1$}}
\put(103.00,5.00){\makebox(0,0)[cc]{$_2$}}
\put(113.00,5.00){\makebox(0,0)[cc]{$_3$}}
\put(123.00,5.00){\makebox(0,0)[cc]{$_4$}}
\put(133.00,5.00){\makebox(0,0)[cc]{$_5$}}
\put(143.00,5.00){\makebox(0,0)[cc]{$_6$}}
\put(153.00,5.00){\makebox(0,0)[cc]{$_7$}}
\put(163.00,5.00){\makebox(0,0)[cc]{$_8$}}
\put(173.00,5.00){\makebox(0,0)[cc]{$_9$}}
\put(183.00,5.00){\makebox(0,0)[cc]{$_a$}}
\put(193.00,5.00){\makebox(0,0)[cc]{$_b$}}
\put(203.00,5.00){\makebox(0,0)[cc]{$_c$}}
\put(213.00,5.00){\makebox(0,0)[cc]{$_0$}}
\put(83.00,15.00){\makebox(0,0)[cc]{$_5$}}
\put(93.00,15.00){\makebox(0,0)[cc]{$_6$}}
\put(103.00,15.00){\makebox(0,0)[cc]{$_7$}}
\put(113.00,15.00){\makebox(0,0)[cc]{$_8$}}
\put(123.00,15.00){\makebox(0,0)[cc]{$_9$}}
\put(133.00,15.00){\makebox(0,0)[cc]{$_a$}}
\put(143.00,15.00){\makebox(0,0)[cc]{$_b$}}
\put(153.00,15.00){\makebox(0,0)[cc]{$_c$}}
\put(163.00,15.00){\makebox(0,0)[cc]{$_0$}}
\put(173.00,15.00){\makebox(0,0)[cc]{$_1$}}
\put(183.00,15.00){\makebox(0,0)[cc]{$_2$}}
\put(192.00,15.00){\makebox(0,0)[cc]{$_3$}}
\put(203.00,15.00){\makebox(0,0)[cc]{$_4$}}
\put(213.00,15.00){\makebox(0,0)[cc]{$_5$}}
\put(163.00,25.00){\makebox(0,0)[cc]{$_5$}}
\put(173.00,25.00){\makebox(0,0)[cc]{$_6$}}
\put(182.00,25.00){\makebox(0,0)[cc]{$_7$}}
\put(192.00,25.00){\makebox(0,0)[cc]{$_8$}}
\put(202.00,25.00){\makebox(0,0)[cc]{$_9$}}
\put(213.00,25.00){\makebox(0,0)[cc]{$_a$}}
\put(192.00,35.00){\makebox(0,0)[cc]{$_0$}}
\put(202.00,35.00){\makebox(0,0)[cc]{$_1$}}
\put(213.00,35.00){\makebox(0,0)[cc]{$_2$}}
\put(93.00,55.00){\makebox(0,0)[cc]{$_0$}}
\put(103.00,55.00){\makebox(0,0)[cc]{$_1$}}
\put(113.00,55.00){\makebox(0,0)[cc]{$_2$}}
\put(123.00,55.00){\makebox(0,0)[cc]{$_3$}}
\put(133.00,55.00){\makebox(0,0)[cc]{$_4$}}
\put(143.00,55.00){\makebox(0,0)[cc]{$_5$}}
\put(153.00,55.00){\makebox(0,0)[cc]{$_6$}}
\put(163.00,55.00){\makebox(0,0)[cc]{$_7$}}
\put(173.00,55.00){\makebox(0,0)[cc]{$_8$}}
\put(183.00,55.00){\makebox(0,0)[cc]{$_9$}}
\put(193.00,55.00){\makebox(0,0)[cc]{$_a$}}
\put(203.00,55.00){\makebox(0,0)[cc]{$_b$}}
\put(213.00,55.00){\makebox(0,0)[cc]{$_c$}}
\put(83.00,55.00){\makebox(0,0)[cc]{$_c$}}
\put(113.00,25.00){\makebox(0,0)[cc]{$_0$}}
\put(123.00,25.00){\makebox(0,0)[cc]{$_1$}}
\put(133.00,25.00){\makebox(0,0)[cc]{$_2$}}
\put(143.00,25.00){\makebox(0,0)[cc]{$_3$}}
\put(153.00,25.00){\makebox(0,0)[cc]{$_4$}}
\put(83.00,25.00){\makebox(0,0)[cc]{$_a$}}
\put(93.00,25.00){\makebox(0,0)[cc]{$_b$}}
\put(103.00,25.00){\makebox(0,0)[cc]{$_c$}}
\put(83.00,35.00){\makebox(0,0)[cc]{$_2$}}
\put(93.00,35.00){\makebox(0,0)[cc]{$_3$}}
\put(103.00,35.00){\makebox(0,0)[cc]{$_4$}}
\put(113.00,35.00){\makebox(0,0)[cc]{$_5$}}
\put(123.00,35.00){\makebox(0,0)[cc]{$_6$}}
\put(133.00,35.00){\makebox(0,0)[cc]{$_7$}}
\put(143.00,35.00){\makebox(0,0)[cc]{$_8$}}
\put(153.00,35.00){\makebox(0,0)[cc]{$_9$}}
\put(163.00,35.00){\makebox(0,0)[cc]{$_a$}}
\put(173.00,35.00){\makebox(0,0)[cc]{$_b$}}
\put(182.00,35.00){\makebox(0,0)[cc]{$_c$}}
\put(83.00,45.00){\makebox(0,0)[cc]{$_7$}}
\put(93.00,45.00){\makebox(0,0)[cc]{$_8$}}
\put(103.00,45.00){\makebox(0,0)[cc]{$_9$}}
\put(113.00,45.00){\makebox(0,0)[cc]{$_a$}}
\put(123.00,45.00){\makebox(0,0)[cc]{$_b$}}
\put(133.00,45.00){\makebox(0,0)[cc]{$_c$}}
\put(143.00,45.00){\makebox(0,0)[cc]{$_0$}}
\put(153.00,45.00){\makebox(0,0)[cc]{$_1$}}
\put(163.00,45.00){\makebox(0,0)[cc]{$_2$}}
\put(173.00,45.00){\makebox(0,0)[cc]{$_3$}}
\put(183.00,45.00){\makebox(0,0)[cc]{$_4$}}
\put(192.00,45.00){\makebox(0,0)[cc]{$_5$}}
\put(203.00,45.00){\makebox(0,0)[cc]{$_6$}}
\put(213.00,45.00){\makebox(0,0)[cc]{$_7$}}
\put(85.00,1.00){\rule{120.00\unitlength}{1.00\unitlength}}
\put(185.00,22.00){\rule{1.00\unitlength}{20.00\unitlength}}
\put(205.00,22.00){\rule{1.00\unitlength}{20.00\unitlength}}
\put(185.00,41.00){\rule{20.00\unitlength}{1.00\unitlength}}
\put(185.00,21.00){\rule{20.00\unitlength}{1.00\unitlength}}
\put(175.00,31.00){\rule{40.00\unitlength}{1.00\unitlength}}
\put(195.00,12.00){\rule{1.00\unitlength}{40.00\unitlength}}
\put(145.00,22.00){\circle{2.00}}
\put(135.00,32.00){\circle{2.00}}
\put(145.00,32.00){\circle{2.00}}
\put(155.00,32.00){\circle{2.00}}
\put(125.00,42.00){\circle{2.00}}
\put(135.00,42.00){\circle{2.00}}
\put(145.00,42.00){\circle{2.00}}
\put(155.00,42.00){\circle{2.00}}
\put(165.00,42.00){\circle{2.00}}
\put(135.00,52.00){\circle{2.00}}
\put(145.00,52.00){\circle{2.00}}
\put(155.00,52.00){\circle{2.00}}
\put(196.00,42.00){\circle*{4.00}}
\put(186.00,42.00){\circle*{4.00}}
\put(196.00,32.00){\circle*{4.00}}
\put(186.00,32.00){\circle*{4.00}}
\put(196.00,22.00){\circle*{4.00}}
\put(186.00,22.00){\circle*{4.00}}
\put(206.00,42.00){\circle*{4.00}}
\put(206.00,32.00){\circle*{4.00}}
\put(206.00,22.00){\circle*{4.00}}
\put(196.00,52.00){\circle*{4.00}}
\put(196.00,12.00){\circle*{4.00}}
\put(175.00,32.00){\circle*{4.00}}
\put(215.00,32.00){\circle*{4.00}}
\put(106.00,2.00){\circle*{4.00}}
\put(96.00,2.00){\circle*{4.00}}
\put(116.00,2.00){\circle*{4.00}}
\put(85.00,2.00){\circle*{4.00}}
\put(125.00,2.00){\circle*{4.00}}
\put(135.00,2.00){\circle*{2.00}}
\put(145.00,2.00){\circle*{2.00}}
\put(155.00,2.00){\circle*{2.00}}
\put(165.00,2.00){\circle*{2.00}}
\put(175.00,2.00){\circle*{2.00}}
\put(156.00,2.00){\circle*{4.00}}
\put(146.00,2.00){\circle*{4.00}}
\put(166.00,2.00){\circle*{4.00}}
\put(135.00,2.00){\circle*{4.00}}
\put(175.00,2.00){\circle*{4.00}}
\put(185.00,2.00){\circle*{2.00}}
\put(195.00,2.00){\circle*{2.00}}
\put(205.00,2.00){\circle*{2.00}}
\put(185.00,2.00){\circle*{2.00}}
\put(195.00,2.00){\circle*{2.00}}
\put(205.00,2.00){\circle*{2.00}}
\put(186.00,2.00){\circle*{4.00}}
\put(196.00,2.00){\circle*{4.00}}
\put(205.00,2.00){\circle*{4.00}}
\emline{160.00}{52.00}{1}{160.00}{52.00}{2}
\emline{86.00}{52.00}{3}{94.00}{52.00}{4}
\emline{96.00}{52.00}{5}{104.00}{52.00}{6}
\emline{106.00}{52.00}{7}{114.00}{52.00}{8}
\emline{116.00}{52.00}{9}{124.00}{52.00}{10}
\emline{126.00}{52.00}{11}{134.00}{52.00}{12}
\emline{136.00}{52.00}{13}{144.00}{52.00}{14}
\emline{146.00}{52.00}{15}{154.00}{52.00}{16}
\emline{156.00}{52.00}{17}{164.00}{52.00}{18}
\emline{166.00}{52.00}{19}{174.00}{52.00}{20}
\emline{176.00}{52.00}{21}{184.00}{52.00}{22}
\emline{186.00}{52.00}{23}{194.00}{52.00}{24}
\emline{196.00}{52.00}{25}{204.00}{52.00}{26}
\emline{206.00}{52.00}{27}{214.00}{52.00}{28}
\emline{160.00}{42.00}{29}{160.00}{42.00}{30}
\emline{86.00}{42.00}{31}{94.00}{42.00}{32}
\emline{96.00}{42.00}{33}{104.00}{42.00}{34}
\emline{106.00}{42.00}{35}{114.00}{42.00}{36}
\emline{116.00}{42.00}{37}{124.00}{42.00}{38}
\emline{126.00}{42.00}{39}{134.00}{42.00}{40}
\emline{136.00}{42.00}{41}{144.00}{42.00}{42}
\emline{146.00}{42.00}{43}{154.00}{42.00}{44}
\emline{156.00}{42.00}{45}{164.00}{42.00}{46}
\emline{166.00}{42.00}{47}{174.00}{42.00}{48}
\emline{176.00}{42.00}{49}{184.00}{42.00}{50}
\emline{186.00}{42.00}{51}{194.00}{42.00}{52}
\emline{196.00}{42.00}{53}{204.00}{42.00}{54}
\emline{206.00}{42.00}{55}{214.00}{42.00}{56}
\emline{160.00}{32.00}{57}{160.00}{32.00}{58}
\emline{86.00}{32.00}{59}{94.00}{32.00}{60}
\emline{96.00}{32.00}{61}{104.00}{32.00}{62}
\emline{106.00}{32.00}{63}{114.00}{32.00}{64}
\emline{116.00}{32.00}{65}{124.00}{32.00}{66}
\emline{126.00}{32.00}{67}{134.00}{32.00}{68}
\emline{136.00}{32.00}{69}{144.00}{32.00}{70}
\emline{146.00}{32.00}{71}{154.00}{32.00}{72}
\emline{156.00}{32.00}{73}{164.00}{32.00}{74}
\emline{166.00}{32.00}{75}{174.00}{32.00}{76}
\emline{176.00}{32.00}{77}{184.00}{32.00}{78}
\emline{186.00}{32.00}{79}{194.00}{32.00}{80}
\emline{196.00}{32.00}{81}{204.00}{32.00}{82}
\emline{206.00}{32.00}{83}{214.00}{32.00}{84}
\emline{160.00}{22.00}{85}{160.00}{22.00}{86}
\emline{86.00}{22.00}{87}{94.00}{22.00}{88}
\emline{96.00}{22.00}{89}{104.00}{22.00}{90}
\emline{106.00}{22.00}{91}{114.00}{22.00}{92}
\emline{116.00}{22.00}{93}{124.00}{22.00}{94}
\emline{126.00}{22.00}{95}{134.00}{22.00}{96}
\emline{136.00}{22.00}{97}{144.00}{22.00}{98}
\emline{146.00}{22.00}{99}{154.00}{22.00}{100}
\emline{156.00}{22.00}{101}{164.00}{22.00}{102}
\emline{166.00}{22.00}{103}{174.00}{22.00}{104}
\emline{176.00}{22.00}{105}{184.00}{22.00}{106}
\emline{186.00}{22.00}{107}{194.00}{22.00}{108}
\emline{196.00}{22.00}{109}{204.00}{22.00}{110}
\emline{206.00}{22.00}{111}{214.00}{22.00}{112}
\emline{160.00}{12.00}{113}{160.00}{12.00}{114}
\emline{86.00}{12.00}{115}{94.00}{12.00}{116}
\emline{96.00}{12.00}{117}{104.00}{12.00}{118}
\emline{106.00}{12.00}{119}{114.00}{12.00}{120}
\emline{116.00}{12.00}{121}{124.00}{12.00}{122}
\emline{126.00}{12.00}{123}{134.00}{12.00}{124}
\emline{136.00}{12.00}{125}{144.00}{12.00}{126}
\emline{146.00}{12.00}{127}{154.00}{12.00}{128}
\emline{156.00}{12.00}{129}{164.00}{12.00}{130}
\emline{166.00}{12.00}{131}{174.00}{12.00}{132}
\emline{176.00}{12.00}{133}{184.00}{12.00}{134}
\emline{186.00}{12.00}{135}{194.00}{12.00}{136}
\emline{196.00}{12.00}{137}{204.00}{12.00}{138}
\emline{206.00}{12.00}{139}{214.00}{12.00}{140}
\emline{85.00}{51.00}{141}{85.00}{43.00}{142}
\emline{85.00}{41.00}{143}{85.00}{33.00}{144}
\emline{85.00}{31.00}{145}{85.00}{23.00}{146}
\emline{85.00}{21.00}{147}{85.00}{13.00}{148}
\emline{85.00}{11.00}{149}{85.00}{4.00}{150}
\emline{95.00}{51.00}{151}{95.00}{43.00}{152}
\emline{95.00}{41.00}{153}{95.00}{33.00}{154}
\emline{95.00}{31.00}{155}{95.00}{23.00}{156}
\emline{95.00}{21.00}{157}{95.00}{13.00}{158}
\emline{95.00}{11.00}{159}{95.00}{4.00}{160}
\emline{105.00}{51.00}{161}{105.00}{43.00}{162}
\emline{105.00}{41.00}{163}{105.00}{33.00}{164}
\emline{105.00}{31.00}{165}{105.00}{23.00}{166}
\emline{105.00}{21.00}{167}{105.00}{13.00}{168}
\emline{105.00}{11.00}{169}{105.00}{4.00}{170}
\emline{115.00}{51.00}{171}{115.00}{43.00}{172}
\emline{115.00}{41.00}{173}{115.00}{33.00}{174}
\emline{115.00}{31.00}{175}{115.00}{23.00}{176}
\emline{115.00}{21.00}{177}{115.00}{13.00}{178}
\emline{115.00}{11.00}{179}{115.00}{4.00}{180}
\emline{125.00}{51.00}{181}{125.00}{43.00}{182}
\emline{125.00}{41.00}{183}{125.00}{33.00}{184}
\emline{125.00}{31.00}{185}{125.00}{23.00}{186}
\emline{125.00}{21.00}{187}{125.00}{13.00}{188}
\emline{125.00}{11.00}{189}{125.00}{4.00}{190}
\emline{135.00}{51.00}{191}{135.00}{43.00}{192}
\emline{135.00}{41.00}{193}{135.00}{33.00}{194}
\emline{135.00}{31.00}{195}{135.00}{23.00}{196}
\emline{135.00}{21.00}{197}{135.00}{13.00}{198}
\emline{135.00}{11.00}{199}{135.00}{4.00}{200}
\emline{145.00}{51.00}{201}{145.00}{43.00}{202}
\emline{145.00}{41.00}{203}{145.00}{33.00}{204}
\emline{145.00}{31.00}{205}{145.00}{23.00}{206}
\emline{145.00}{21.00}{207}{145.00}{13.00}{208}
\emline{145.00}{11.00}{209}{145.00}{4.00}{210}
\emline{155.00}{51.00}{211}{155.00}{43.00}{212}
\emline{155.00}{41.00}{213}{155.00}{33.00}{214}
\emline{155.00}{31.00}{215}{155.00}{23.00}{216}
\emline{155.00}{21.00}{217}{155.00}{13.00}{218}
\emline{155.00}{11.00}{219}{155.00}{4.00}{220}
\emline{165.00}{51.00}{221}{165.00}{43.00}{222}
\emline{165.00}{41.00}{223}{165.00}{33.00}{224}
\emline{165.00}{31.00}{225}{165.00}{23.00}{226}
\emline{165.00}{21.00}{227}{165.00}{13.00}{228}
\emline{165.00}{11.00}{229}{165.00}{4.00}{230}
\emline{175.00}{51.00}{231}{175.00}{43.00}{232}
\emline{175.00}{41.00}{233}{175.00}{33.00}{234}
\emline{175.00}{31.00}{235}{175.00}{23.00}{236}
\emline{175.00}{21.00}{237}{175.00}{13.00}{238}
\emline{175.00}{11.00}{239}{175.00}{4.00}{240}
\emline{185.00}{51.00}{241}{185.00}{43.00}{242}
\emline{185.00}{21.00}{243}{185.00}{13.00}{244}
\emline{185.00}{11.00}{245}{185.00}{4.00}{246}
\emline{196.00}{11.00}{247}{196.00}{4.00}{248}
\emline{205.00}{51.00}{249}{205.00}{43.00}{250}
\emline{205.00}{21.00}{251}{205.00}{13.00}{252}
\emline{205.00}{11.00}{253}{205.00}{4.00}{254}
\emline{215.00}{51.00}{255}{215.00}{43.00}{256}
\emline{215.00}{41.00}{257}{215.00}{33.00}{258}
\emline{215.00}{31.00}{259}{215.00}{23.00}{260}
\emline{215.00}{21.00}{261}{215.00}{13.00}{262}
\emline{215.00}{11.00}{263}{215.00}{4.00}{264}
\emline{206.00}{2.00}{265}{214.00}{2.00}{266}
\end{picture}
\end{figure}

\noindent However, for our purpose, we will utilize an \textquotedblleft
inverse\textquotedblright\ process. Given an induced subgraph $D=(V,E)$ of $%
\Lambda _{n},$ find a partition (tiling) of $\Lambda _{n}$ into copies of $D$%
. Here we mean partitioning of the vertex set of $\Lambda ${\ only, see
Remark }\ref{RR}. Hence we need to find a suitable lattice $L$ that would
allow the required choice of the set $T$, i.e. $[T]=D$. It turns out that to
do so one does not have to find the lattice $L$ explicitly. We will show
that the following construction leads to the desired tiling of ${\Lambda }%
_{n}$. We claim that if there exists an Abelian group $(G,+)$ of order $|V|$
and elements $g_{1},\ldots ,g_{n}$ of $G$ such that the restriction of the
homomorphism $\Phi :\mathbb{Z}^{n}\rightarrow G,$ $\Phi ((a_{1},\ldots
,a_{n}))=a_{1}\Phi (e_{1})+\ldots +a_{n}\Phi (e_{n})=a_{1}g_{1}+\ldots
+a_{n}g_{n},$ to $V$ is a bijection then there exists a partition of $%
\Lambda _{n}$ into copies of $D$. In other words, we need to find an Abelian
group $G$ of order $|V|$ and assign elements $g_{1},\ldots ,g_{n}$ of $G$ to
the vertices $e_{1},\ldots ,e_{n}$ of $\Lambda _{n}$ so that $\Phi
((a_{1},\ldots ,a_{n}))=a_{1}\Phi (e_{1})+\ldots +a_{n}\Phi
(e_{n})=a_{1}g_{1}+\ldots +a_{n}g_{n},$ is a bijection on $V$. It is well
known, that the ker of a homomorphism $\phi :A\rightarrow B$ is a subgroup
of $A.$ Thus, the elements $w$ of $\mathbb{Z}^{n}$ for which $\Phi (w)=0$
form a lattice $L$ in $(\mathbb{Z}^{n},+)$. In addition, $(\mathbb{Z}%
^{n},+)/L=G$ and the vertex set $V$ comprises exactly one element from each
coset of $(\mathbb{Z}^{n},+)/L$; thus we can set $T=V$.\bigskip

\noindent As the above method is the main tool in this paper, we summarize
it as Corollary \ref{C} (to Theorem \ref{BBB} below)

\begin{theorem}
\rm{\cite{Hnew}} \label{BBB}Let $D=(V,E)$ be a subgraph of $\Lambda _{n}$. Then
there is a lattice-like tiling of $\Lambda _{n}$ by copies of $D$ if and
only if there is an Abelian group $(G,\circ )$ and a homomorphism $\Phi :%
\mathbb{Z}^{n}\rightarrow G$, so that the restriction of $\Phi $ to $V$ is a
bijection.
\end{theorem}

\noindent If the restriction of $\Phi $ to $V$ is an injection, then Theorem %
\ref{BBB} (in which $D$ need not be connected) produces a packing of $%
\Lambda _{n}$ by copies of $D$. This idea has been used in several papers,
see e.g. \cite{S,HS,Sw}. The following corollary of Theorem \ref{BBB} is
tailored to our present needs:

\begin{corollary}
\label{C}Let $t\geq 1$ and let $H$ be a subgraph of $\Lambda _{n}$. Further,
let $H^{\ast }$ be an induced supergraph of $H$ such that a vertex $v$
belongs to $H^{\ast }$ if and only if $d(v,H)\leq t$; let $D=(V,E)$ be a
copy of $H^{\ast }$ or a copy of a disjoint union of finitely many copies of
$H^{\ast }$ that contains vertices $O,e_{1},\ldots ,e_{n}$. Then, there is a
$t$-$\mathrm{PDDS}[H]$ if there exists an Abelian group $G$ of order $|V|$
and a homomorphism $\Phi :\mathbb{Z}^{n}\rightarrow G$ such that the
restriction of $\Phi $ to $V$ is a bijection.
\end{corollary}

\begin{remark}
\noindent We will always choose $D$ to contain vertices $O,e_{1},\ldots
,e_{n}.$ This is not a necessary condition but it will be added to simplify
the exposition. A $t$-$\mathrm{PDDS}[H]$ constructed by means of Corollary %
\ref{C} is lattice-like if $D$ is isomorphic to $H^{\ast }$. If $D$ consists
of more copies of $H^{\ast },$ then we get a lattice tiling of $\mathbb{Z}%
^{n}$ by $D$ but this will not constitute a lattice-like $t$-$\mathrm{PDDS}%
[H].$
\end{remark}

\noindent The rest of the paper is organized as follows. Section 2
contains a proof of Theorem \ref{T}, while a proof of
$\mathrm{\mathbf{(i)}}$ of Theorem \ref{T1} will be given in Section
3. Theorem \ref{T2} will be proved in Section 4. To demonstrate the
strength of the construction, in Section 5 we present a periodic
$1$-PDDS in $\Lambda _{2}$ that is not lattice-like.

\section{Existence of $t$-$\mathrm{PDDS}$\thinspace s}

\label{sec:1}

\noindent In this section we prove Theorem \ref{T}, that is we prove the
existence of $t$-$\mathrm{PDDS}$\thinspace s as described in Conjecture \ref%
{N}. For the sake of completeness we note that a Minkowski's tiling that
proves part \textrm{(v)}\textit{\ }can be obtained by Corollary \ref{C}
using the group $G=\mathbb{Z}_{38}$ and the homomorphism given by $\Phi
(e_{1})=1$, $\Phi (e_{2})=11$ and $\Phi (e_{3})=7$.

\subsection{Part (i)}

\noindent Here we deal with the case when each component of a 1-$\mathrm{PDDS%
}$ is isomorphic to a path $P_{k}$ of length $k-1$, where $k\geq 2$. We
start with the case when each component of a $t$-$\mathrm{PDDS}$ is an
isolated vertex. Each $1$-$\mathrm{PDDS}[P_{1}]$ in $\Lambda _{n}$
corresponds to a perfect 1-error correcting Lee code, $\mathrm{PLC}(n,1)$.
The existence of such codes has been showed independently by several
authors. K\'{a}rteszi asked whether there exists a $\mathrm{PLC}(3,1)$.
Feller, for $n=3,$ and then Korchm\'{a}ros, and Golomb and Welch \cite{GW}
showed that there is a $\mathrm{PLC}(n,1)$ for all $n\geq 2 $. The following
stronger theorem has been proved by Moln\'{a}r \cite{Molnar}.

\begin{theorem}
\label{M}The number of non-congruent lattice-like $\mathrm{PLC}(n,1)$ codes equals
the number of Abelian groups of order $2n+1$.
\end{theorem}

\noindent To illustrate our method we prove the theorem. The following proof
is shorter than the original one due to Moln\'{a}r. Since in this case $H$
is an isolated vertex, the graph $H^{\ast }$ is of order $2n+1$. We choose a
copy of $D=(V,E)$ of $H^{\ast }$ such that $V=\{\pm e_{i}\,;\,i=1,\ldots
,n\}\cup \{O\}$. Let $G$ be an Abelian group of order $2n+1$. Choose a set $%
K=\{g_{1},\ldots ,g_{n}\}$ formed by $n$ distinct elements of $G$ such that $%
K$ contains exactly one element from each pair $g,g^{-1}$; formally, $g\in K$
if and only if $g^{-1}\notin K$. Since no element of $G$ is of order $2$,
the set $K$ is well defined. Clearly, the restriction of the homomorphism $%
\Phi :\mathbb{Z}^{n}\rightarrow G$ given by $\Phi ((a_{1},\ldots
,a_{n}))=\Phi (e_{1})^{a_{1}}\circ \ldots \circ \Phi (e_{n})^{a_{n}}$ to $V$
is a bijection. Thus, each Abelian group of order $2n+1$ generates a $%
\mathrm{PLC}(n,1)$; this code is a periodic code where $p_{i}\,$s are orders
of elements of $G.$ It is not difficult to check that non-isomorphic groups
generate non-congruent $\mathrm{PLC}(n,1)$ codes.\bigskip

\noindent We note that Szab\'{o} \cite{Sz} constructed, in the case when $%
2n+1$ is not a prime, the first non-lattice-like $\mathrm{PLC}(n,1)$ code.
This code is periodic though. In \cite{HB}, for the same case, the first
non-periodic $\mathrm{PLC}(n,1)$ code has been found. It has also been shown
in \cite{HB} that there is a unique $\mathrm{PLC}(n,1)$ code for $n=2,3$%
.\bigskip

\noindent The existence of $1$-$\mathrm{PDDS}[P_{2}]$ (called total perfect
codes in \cite{KG}) has been proved in \cite{E} in terms of diameter perfect
codes.

\begin{theorem}
A $1$-$\mathrm{PDDS}[P_{k}]$ in $\Lambda _{n}$ exists for each $n\geq 2$ and
each $k\geq 1$.
\end{theorem}

\begin{proof}
We will construct the desired $\mathrm{PDDS}$ by applying Corollary \ref{C}.
Set $H=P_{k}$. We place the graph $D=(V,E)$ that is isomorphic to $H^{\ast }$
in such a way that $V$ comprises the vertices $O$, $e_{1}$, $2e_{1}$, $%
\ldots $, $(k-1)e_{1}$ of the path $P_{k}$ and their $2nk-2k+2$ neighbors,
namely $-e_{1}$, $ke_{1}$ and $\pm e_{i}$, $e_{1}\pm e_{i}$, $\ldots $, $%
(k-1)e_{1}\pm e_{i}$ for $i=2,\ldots ,n$. Thus, $|V|=2nk-k+2$ and $D$
contains the vertices $O$ and $e_{i}$, for $i=1,\ldots ,n$, as required by
Corollary \ref{C}. We choose $G=\mathbb{Z}_{2nk-k+2}$. The element $g_{i}$
of $G$ that is assigned to the vertex $e_{i}$, for $i=1,\ldots ,n$, is $%
g_{i}=(i-1)k+1$. To finish the proof, we need to show that the restriction
of the mapping $\Phi ((a_{1},\ldots ,a_{n}))=\Phi (e_{1})^{a_{1}}\circ
\ldots \circ \Phi (e_{n})^{a_{n}}=a_{1}g_{1}+\ldots +a_{n}g_{n}$ to the set $%
V$ is a bijection. To see this, it suffices to note that $\Phi
\{O,e_{1},2e_{1},\ldots ,(k-1)e_{1}\}=\{0,1,\ldots ,k-1\}$, $\Phi
\{-e_{1},ke_{1}\}=\{k,2nk-k+1\},$ and $\Phi \{\pm e_{i},e_{1}\pm
e_{i},\ldots ,(k-1)e_{1}\pm e_{i}\}=\{\pm (i-1)k+1,\pm (i-1)k+2,\ldots ,\pm
(i-1)k+k-1,\pm ik\}$. In aggregate, $\Phi (V)=\{0,\ldots ,k\}\cup $ $%
\bigcup\limits_{i=2}^{n}\{(i-1)k+1,\ldots ,ik\}\cup
\bigcup\limits_{i=2}^{n}\{(2n-i)k+1,\ldots ,(2n-i-1)k+2\}\cup
\{2nk-k+1\}=\{0,\ldots ,2nk-k+1\}=G$. For the reader convenience we
illustrate the proof by means of three small examples for $k=3$:

$$\begin{array}{||l||ccccc|ccccc|ccccc||}\hline
&&<&e_1,e_2&>&&&<&e_1,e_3&>&&&<&e_1,e_4&>&\\\hline\hline
_{n=2}&&_7&_8&_9&&&&&&&&&&&\\
_{\mathbb{Z}_{11}}&^{10}&^\mathbf0_4&^\mathbf1_5&^\mathbf2_6&^3&&&&&&&&&&\\\hline
_{n=3}&&_{13}&_{14}&_{15}&&&_{10}&_{11}&_{12}&&&&&&\\
_{\mathbb{Z}_{17}}&^{16}&^{\,\,\,\mathbf0}_{\,\,\,4}&^{\,\,\,\mathbf1}_{\,\,\,5}&^{\,\,\,\mathbf2}_{\,\,\,6}&^{\,\,\,3}&^{16}&^{\,\,\,\mathbf0}_{\,\,\,7}&^{\,\,\,\mathbf1}_{\,\,\,8}&^{\,\,\,\mathbf2}_{\,\,\,9}&^3&&&&&\\\hline
_{n=4}&&_{19}&_{20}&_{21}&&&_{16}&_{17}&_{18}&&&_{13}&_{14}&_{15}&\\
_{\mathbb{Z}_{23}}&^{22}&^{\,\,\,\mathbf0}_{\,\,\,4}&^{\,\,\,\mathbf1}_{\,\,\,5}&^{\,\,\,\mathbf2}_{\,\,\,6}&^{\,\,\,3}&^{22}&^{\,\,\,\mathbf0}_{\,\,\,7}&^{\,\,\,\mathbf1}_{\,\,\,8}&^{\,\,\,\mathbf2}_{\,\,\,9}&^{\,\,\,3}&^{22}&^{\,\,\,\mathbf0}_{10}&^{\,\,\,\mathbf1}_{11}&^{\,\,\,\mathbf2}_{12}&^3\\\hline
\end{array}$$
\end{proof}

\subsection{Part (ii)}

\noindent In this subsection we prove the existence of a $t$-$\mathrm{PDDS}$
in $\Lambda _{2}$ whose components are all isomorphic to a path $P_{k}$,
where $t>1$ and $k>1$.

%\newpage

\begin{theorem}
A $t$-$\mathrm{PDDS}[P_{k}]$ in $\Lambda _{2}$ exists for each $t\geq 1$ and
$k\geq 1$.
\end{theorem}

\begin{proof}
We provide a detailed proof as we use the same approach to prove this and
the next theorem. Let $H$ be a path $P_{k}$ on vertices $\{O,e_{2},2e_{2},%
\ldots ,(k-1)e_{2}\}.$ Then $H^{\ast }$ consists of vertices of $H$ plus all
vertices at distance at most $t$ from $H;$ hence $\left\vert H^{\ast
}\right\vert =$ $2t^{2}+2tk+k.$ Clearly, $xe_{1}+ye_{2}\in H^{\ast }$ iff
\begin{eqnarray*}
-t &\leq &x<0\mbox{ and }-x-t\leq y\leq x+t+k-1 \\
&&\mbox{or} \\
0 &\leq &x\leq t\mbox{ and }x-t\leq y\leq -x+t+k-1
\end{eqnarray*}

\noindent We will construct the desired $\mathrm{PDDS}$ by applying
Corollary \ref{C} so that the graph $D=(V,E)$ consists of two disjoint
copies of $H^{\ast };$ a copy described above and a translation of this copy
by $(t,t+k).$ Thus, the other copy of $H^{\ast }$ is given by%
\begin{eqnarray*}
0 &\leq &x\leq t\mbox{ and }-x+t+k\leq x+t+2k-1 \\
&&\mbox{or} \\
t+1 &\leq &x\leq 2t\mbox{ and }x-t+k\leq y\leq -x+3t+2k-1
\end{eqnarray*}

\noindent In aggregate, $\left\vert V\right\vert =4t^{2}+4tk+2k,$ and a
vertex $xe_{1}+ye_{2}\in V$ iff%
\begin{eqnarray}
-t &\leq &x<0\mbox{ and }-x-t\leq y\leq x+t+k-1  \label{NN} \\
&&\mbox{either}  \\
0 &\leq &x\leq t\mbox{ and }x-t\leq y\leq x+t+2k-1   \\
&&\mbox{or}  \\
t+1 &\leq &x\leq 2t\mbox{ and }x-t+k\leq y\leq -x+3t+2k-1
\end{eqnarray}

\noindent To construct the desired lattice-like PDDS we choose
the cyclic group $G=$

\noindent$\mathbb{Z}_{4t^{2}+4tk+2k}$ and set $g_{1}=2t+2k-1$, and $g_{2}=1.$
Hence $\Phi (xe_{1}+ye_{2})=((2t+2k-1)x+y)$ mod $(4t^{2}+4tk+2k).$\bigskip

\noindent For fixed $x,$ by (\ref{NN}), the set $I_{x}=\{y;$ $%
xe_{1}+ye_{2}\in V\}$ is an interval. Therefore, as $g_{2}=1,\Phi (I_{x})$
comprises $\left\vert I_{x}\right\vert $ consecutive elements of the group $%
G,$ where we take that $0$ follows the element $4t^{2}+4tk+2k-1.$ To see
that the mapping $\Phi $ is a bijection on $V$ it is sufficient to show that
the intervals $I_{x},-t\leq x\leq 2t$ can be ordered in such a way that if $%
I_{z}$ immediately precedes $I_{v}$in this order then $\Phi (\min
I_{v})=\Phi (\max I_{z})+1$. An order with this property is given implicitly
below.\bigskip

\noindent $\mathrm{\mathbf{(i)}}$ for each $-t\leq x\leq 0,$ it is $\Phi
(\min I_{x})=\Phi (\max I_{x+2t})+1;$

\noindent $\mathrm{\mathbf{(ii)}}$ for each $1\leq x\leq t,$ it is $\Phi
(\min I_{x})=\Phi (\max I_{x-1})+1;$

\noindent $\mathrm{\mathbf{(iii)}}$ for each $t+1\leq x\leq 2t,$ it is $\Phi
(\min I_{x})=\Phi (\max I_{-2t-1+x})+1.$\bigskip

\noindent It is easy to prove $\mathrm{\mathbf{(i)}}$-$\mathrm{\mathbf{(iii)}%
}$ by using (\ref{NN}) and simple calculations. For the readers convenience
we work out details of $\mathrm{\mathbf{(i)}}$. If $-t\leq x\leq 0,$ then,
from the first line of (\ref{NN}), $\Phi (\min I_{x})=\Phi
(xe_{1}+(-x-t)e_{2})=(x(2t+2k-1)+(-x-t))$ mod $(4t^{2}+4tk+2k)=$

\noindent$(2(t+k-1)x-t)$ mod $(4t^{2}+4tk+2k).$ \bigskip

\noindent For $-t+1\leq x\leq 0,$ by the third line of (\ref{NN}), we get

\noindent$\Phi (\max I_{x+2t})=\Phi ((x+2t)e_{1}+(-(x+2t)+3t+2k-1)e_{2})=$

\noindent$((x+2t)((2t+2k-1)+(-(x+2t)+3t+2k-1))$ mod $(4t^{2}+4tk+2k)=$

\noindent$([2(t+k-1)x-t]+[4t^{2}+4tk+2k]-1)$ mod $%
(4t^{2}+4tk+2k)=([2(t+k-1)x-t]-1)$ mod $(4t^{2}+4tk+2k)=$ $\Phi (\min
I_{x})-1.$\bigskip

\noindent Finally, for $x=-t,$ by the second line of (\ref{NN}), $\Phi (\max
I_{x+2t})=\Phi ((x+2t)e_{1}+(x+2t+t+2k-1)e_{2})=$

\noindent$(t(2t+2k-1)+(2t+2k-1))$ mod $(4t^{2}+4tk+2k)=$

\noindent $([2(t+k-1)(-t)-t+[4t^{2}+4tk+2k]-1)$ mod $%
(4t^{2}+4tk+2k)=(2(t+k-1)(-t)-t)$ mod $(4t^{2}+4tk+2k)=$ $\Phi (-t)-1.$ The
proof is complete.\bigskip

\noindent For the reader's convenience, we provide two small examples for $%
t=2,3$ and $k=3.$

$$
\begin{array}{||lllllll||llllllllllll||}
 & _{36} & _{45}^{44} & _{\,\,8} &  &  &  &  &
 &  & _{65} & _{76}^{75} & _{\,\,9} &  &  &
&  &  &  \\
_{29}^{28} & _{38}^{37} & _{\,\,\mathbf{1}}^{\,\,\mathbf{0}} &
_{10}^{\,\,9} & _{19}^{18} &  &  &  & _{45} &
_{56}^{55} & _{67}^{66} & _{\,\,\mathbf{0}}^{77} & _{11}^{10}
& _{22}^{21} & _{33} &  &  &  &  \\
^{30} & _{40}^{39} & _{\,\,3}^{\,\,\mathbf{2}} & _{12}^{11} &
_{21}^{20} &  &  &  & _{47}^{46} & _{58}^{57} &
_{69}^{68} & _{\,\,\mathbf{2}}^{\,\,\mathbf{1}} & _{13}^{12} &
_{24}^{23} & _{35}^{34} &  &  &  &  \\
 &  & _{\,\,5}^{\,\,4} & _{14}^{13} & _{\mathbf{23}}^{22}
& _{32}^{31} & _{41} &  &  & ^{59} & _{71}^{70} &
_{\,\,4}^{\,\,3} & _{15}^{14} & _{26}^{25} & _{37}^{36} &
_{48} &  &  &  \\
 &  & _{\,\,7}^{\,\,6} & _{16}^{15} & _{\mathbf{25}}^{%
\mathbf{24}} & _{34}^{33} & _{43}^{42} &  &  &  &
& _{\,\,6}^{\,\,5} & _{17}^{16} & _{28}^{27} & _{\mathbf{39}%
}^{38} & _{50}^{49} & _{61}^{60} & _{72} &  \\
 &  &  & ^{17} & _{27}^{26} & ^{35} &  &  &
 &  &  & _{\,\,8}^{\,\,7} & _{19}^{18} & _{30}^{29}
& _{\mathbf{41}}^{\mathbf{40}} & _{52}^{51} & _{63}^{62} &
_{74}^{73} &  \\
 &  &  &  &  &  &  &  &  &  &  &
 & ^{20} & _{32}^{31} & _{43}^{42} & _{54}^{53} &
^{64} &  &  \\
 &  &  &  &  &  &  &  &  &  &  &
 &  &  & ^{44} &  &  &  &  \\
&  &  &  &  &  &  &  &  &  &  &  &  &  &  &  &  &  &
\end{array}%
$$

\noindent To prove the statement of this Theorem 8 just with $D=(V,E)=H^*$,
notice that now $|V|=2t^{2}+2tk+k$ and choose the cyclic group $G=\mathbb{Z}%
_{2t^{2}+2tk+k}$, setting $g_{1}=1$ and $g_2=2t+1.$ Hence $\Phi
(xe_{1}+ye_{2})=(x+(t+1)y)$ mod $(2t^{2}+2tk+k)$ and $\Phi$ maps $V$
bijectively onto $G$ by sending the successive intersections of $V$ with the
lines $e_2=0,\ldots,r,-t,r+1,-t+1,r+2,\ldots,-1,r+t$ from left to right onto
$-tg_1,\ldots,-g_1,O,\ldots,(|V|-t)g_1$. For the reader's convenience, we
provide two small examples for $t=2,3$ and $k=3.$
$$\begin{array}{||rrrrr||rrrrrrr||}
^{  }_{  }&^{  }_{17}&^{13}_{18}&^{  }_{19}&^{  }_{  }&^{  }_{  }&^{  }_{  }&^{  }_{24}&^{18}_{25}&^{  }_{26}&^{  }_{  }&^{  }_{  }\\
^{21}_{ 3}&^{22}_{ 4}&^{\bf0}_{\bf5}&^{ 1}_{ 6}&^{ 2}_{ 7}&^{  }_{36}&^{30}_{37}&^{31}_{38}&^{32}_{\,\,\,\bf0}&^{33}_{\,\,\,1}&^{34}_{\,\,\,2}&^{  }_{\,\,\,3}\\
^{ 8}_{  }&^{ 9}_{14}&^{\bf{10}}_{15}&^{11}_{16}&^{12}_{  }&^{ 4}_{11}&^{ 5}_{12}&^{ 6}_{13}&^{\,\,\,\bf7}_{\bf{14}}&^{\,\,\,8}_{15}&^{\,\,\,9}_{16}&^{10}_{17}\\
^{  }_{  }&^{  }_{  }&^{20}_{  }&^{  }_{  }&^{  }_{  }&^{  }_{  }&^{19}_{  }&^{20}_{27}&^{21}_{28}&^{22}_{29}&^{23}_{  }&^{  }_{  }\\
          &          &          &          &          &^{  }_{  }&^{  }_{  }&^{  }_{  }&^{35}_{  }&^{  }_{  }&^{  }_{  }&^{  }_{  }\\
\end{array}$$
\end{proof}

\subsection{Part (iii)}

\noindent Here we discuss the existence of a $t$-$\mathrm{PDDS}$ in $\Lambda
_{2}$ whose components are isomorphic to the Cartesian product of two finite
paths. The case $k=1$ of the following theorem, using a different technique,
has been also proved in \cite{E} in terms of diameter perfect codes.

\begin{theorem}
\label{abc}A $t$-$\mathrm{PDDS}$ in $\Lambda _{2}$ whose components are
isomorphic to $P_{2}\square P_{k}$ exists for each $t\geq 1$ and $k\geq 1$.
\end{theorem}

\begin{proof}
We prove this theorem using the same approach as in Theorem 8 and indicate
at the end how to obtain the same result just with $D=H^*$. Let $H $ be the
graph $P_{2}\square P_{k}$ on vertices $\{re_{2},e_{1}+re_{2};0\leq r\leq
k-1\}.$ Then the graph $H^{\ast }$ consisting of $H$ and all vertices at
distance at most $t$ from $H$ is of order $2t^{2}+2tk+2t+2k$. It is easy to
see that $xe_{1}+ye_{2}\in H^{\ast }$ iff%
\begin{eqnarray*}
-t &\leq &x\leq 0\mbox{ and }-x-t\leq y\leq x+k+t-1 \\
&&\mbox{or} \\
1 &\leq &x\leq t+1\mbox{ and }x-t-1\leq y\leq -x+k+t
\end{eqnarray*}

\noindent We will construct the desired $\mathrm{PDDS}$ by applying
Corollary \ref{C} to the graph $D=(V,E)$ consisting of two disjoint copies
of $H^{\ast };$ a copy described above and a translation of this copy by $%
(t+1,t+k).$ Thus, the other copy of $H^{\ast }$ is given by $%
xe_{1}+ye_{2}\in H^{\ast }$ iff
\begin{eqnarray*}
1 &\leq &x\leq t+1\mbox{ and }-x+t+k+1\leq y\leq x+2k+t-2 \\
&&\mbox{or} \\
t+2 &\leq &x\leq 2t+2\mbox{ and }x+k-t-2\leq y\leq -x+2k+3t+1
\end{eqnarray*}

\noindent In aggregate, a vertex $xe_{1}+ye_{2}\in V$ iff

\begin{eqnarray}
-t &\leq &x\leq 0\mbox{ and }-x-t\leq y\leq x+k+t-1  \label{NNN} \\
&&\mbox{or}   \\
1 &\leq &x\leq t+1\mbox{ and }x-t-1\leq y\leq x+2k+t-1  \\
&&\mbox{or}   \\
t+2 &\leq &x\leq 2t+2\mbox{ and }x+k-t-2\leq y\leq -x+2k+3t+1
\end{eqnarray}

\noindent To construct the desired lattice-like PDDS we choose
the Abelian group $G=\mathbb{Z}_{2t+2k}\times $ $\mathbb{Z}_{2t+2}$ and set $%
g_{1}=(0,1)$, and $g_{2}=(1,0).$ Hence $\Phi (xe_{1}+ye_{2})=(x$ mod $%
(2t+2k),y$ mod $(2t+2)).$ To finish the proof we show that a restriction of $%
\Phi $ to $V$ is a bijection. Let, as above, $I_{x}=\{y;xe_{1}+ye_{2}\in
V\}. $ Then, for all $1\leq x\leq t+1,$ $\Phi (I_{x})=\mathbb{Z}%
_{2t+2k}\times \{x\},$ as $g_{2}=(1,0)$ and $I_{x}$ is an interval of length
$2t+2k.$

\noindent Now, for all $t+2\leq x\leq 2t+2,$ it suffices to realize that

\noindent $I_{x}\cup I_{x-(2t+2)}=[(-x+(2t+2)-t,x-(2t+2)+k+t-1]\cup \lbrack
x+k-t-2,-x+2k+3t+1]=$

\noindent $\lbrack -x+t+2,x-t+k-3]\cup \lbrack
x-t+k-2,-x+2k+3t+1]=[-x+t+2,-x+2k+3t+1].$

\noindent Thus, $I_{x}\cup I_{x-(2t+2)}$ is an interval of length $2t+2k$ as
well$.$ This in turn implies, as $x\equiv x-(2t+2)$ mod $(2t+2),$ that $\Phi
(I_{x}\cup I_{x-(2t+2)})=\mathbb{Z}_{2t+2k}\times \{x\}$ also in this case.
The proof is complete. However, after a pair of examples, we say how to make
out with $D=H^*$. \bigskip

\noindent For the reader's convenience, we illustrate the proof with some
small examples. For $t=2$ and $k=1,2$, we take $G=\mathbb{Z}_{4+2k}\times%
\mathbb{Z}_6$ and $\Phi$ assigned as follows:
$$
\begin{array}{||ccccccccc||ccccccccc||}
 & _{5,5} & _{5,0}^{4,0} & _{5,1}^{4,1} & _{5,2} &  &
 &  &  &  & _{7,5} & _{7,0}^{6,0} & _{7,1}^{6,1}
& _{7,2} &  &  &  &  \\
^{0,4} & _{1,5}^{0,5} & _{1,0}^{\mathbf{0,0}} & _{1,1}^{\mathbf{%
0,1}} & _{1,2}^{0,2} & _{1,3}^{0,3} & _{1,4} &  &  &
_{1,4}^{0,4} & _{1,5}^{0,5} & _{\mathbf{1,0}}^{\mathbf{0,0}} &
_{\mathbf{1,1}}^{\mathbf{0,1}} & _{1,2}^{0,2} & _{1,3}^{0,3} &
 &  &  \\
 &  & ^{2,0} & _{3,1}^{2,1} & _{3,2}^{2,2} & _{\mathbf{%
3,3}}^{2,3} & _{\mathbf{3,4}}^{2,4} & _{3,5}^{2,5} & _{3,6} &
 & ^{2,5} & _{3,0}^{2,0} & _{3,1}^{2,1} & _{3,2}^{2,2} &
_{3,3}^{2,3} & _{3,4}^{2,4} & _{3,5} &  \\
 &  &  &  & ^{4,2} & _{5,3}^{4,3} & _{5,4}^{4,4} &
^{4,5} &  &  &  &  & _{5,1}^{4,1} & _{5,2}^{4,2}
& _{\mathbf{5,3}}^{\mathbf{4,3}} & _{\mathbf{5,4}}^{\mathbf{4,4}} &
_{5,5}^{4,5} & _{5,6}^{4,6} \\
 &  &  &  &  &  &  &  &  &  &  &
 &  & ^{6,2} & _{7,3}^{6,3} & _{7,4}^{6,4} & ^{6,5}
&  \\
&  &  &  &  &  &  &  &  &  &  &  &  &  &  &  &  &
\end{array}%
$$

\noindent To prove the statement of this Theorem 9 just with $D=(V,e)=H^*$,
note that $|V|=2(t+1)(t+k)$ and denote $m=\gcd(t+1,t+k)$. Then take:

\begin{enumerate}
\item $G=\mathbb{Z}_{2(t+1)(t+k)}$, $g_1=t+1$ and $g_2=t+k$, if $m=1$;

\item $G =\mathbb{Z}_m\times\mathbb{Z}_n$, where $n =\frac{2(t+1)(t+k)}{m}$%
\hspace{10mm}, if $m \neq 1$; now take:

\begin{enumerate}
\item $g_1=(1,n)$\hspace{8mm} and $g_2=(0,1)$\hspace{11mm}, if $m|t+k$;

\item $g_1=(1,\frac{n}{2(2t+1)})$ and $g_2=(1,\frac{2t+1}{m})$\hspace{6mm},
otherwise.
\end{enumerate}
\end{enumerate}

\noindent We leave the details of the proof of this approach of Theorem 9
to the reader and just give three small examples of it, for $(t,k)=(2,2)$, $%
(2,4)$, $(3,3)$, where $G=\mathbb{Z}_{24}$, $\mathbb{Z}_3\times\mathbb{Z}%
_{12}$, $\mathbb{Z}_2\times\mathbb{Z}_{24}$, respectively:
$$\begin{array}{cccccc||cccccc||cccccccc}
&_{11}&^{\,\,6}_{15}&^{10}_{19}&_{23}&&&_{2,9}&^{\,0,10}_{\,0,11}&^{1,0}_{1,1}&_{2,3}&&&&_{1,17}&^{1,18}_{0,20}&^{0,21}_{1,23}&_{0,2}&&\\
^{16}_{\,\,1}&^{20}_{\,\,5}&^{\bf0}_{\bf9}&^{\,\,\bf4}_{\bf{13}}&^{\,\,8}_{17}&^{12}_{21}&^{1,8}_{1,9}&^{2,10}_{2,1}&^{\bf{0,0}}_{\bf{0,1}}&^{\bf{1,2}}_{\bf{1,3}}&^{2,4}_{2,5}&^{0,6}_{0,7}&_{1,15}&^{1,16}_{0,18}&^{0,19}_{1,21}&^{1,22}_{\bf{0,0}}&^{0,1}_{\bf{1,3}}&^{1,4}_{0,6}&^{0,7}_{1,9}&_{0,12}\\
&^{14}&^{18}_{\,\,3}&^{22}_{\,\,7}&^{26}&&^{1,10}_{1,11}&^{2,2}_{2,3}&^{\bf{0,2}}_{\bf{0,3}}&^{\bf{1,4}}_{\bf{1,5}}&^{2,6}_{2,7}&^{0,8}_{0,9}&^{0,17}_{1,19}&^{1,20}_{0,22}&^{0,23}_{1,1}&^{\bf{1,2}}_{\bf{0,4}}&^{\bf{0,5}}_{\bf{1,7}}&^{1,8}_{0,10}&^{0,11}_{1,13}&^{1,14}_{0,16}\\
&&&&&&&^{2,4}&^{0,4}_{0,5}&^{1,6}_{1,7}&^{2,8}&&&^{1,0}&^{0,3}_{1,5}&^{1,6}_{0,8}&^{0,9}_{1,11}&^{1,12}_{0,14}&^{0,15}&\\
&&&&&&&&&&&&&&&^{1,10}&^{0,13}&&&\\
\end{array}$$
\end{proof}

\subsection{Part (iv)}

\noindent In this subsection we discuss the existence of $t$-$\mathrm{PDDS}$%
\thinspace s in $\Lambda _{n}$ whose components are isomorphic to $%
P_{2}\square P_{2}$. Note that for $n=2$ this case overlaps with the
previous part.

\begin{theorem}
Let $n=3k+2$, where $k\geq 0$. Then, there exists a lattice-like $1$-$%
\mathrm{PDDS}$ in $\Lambda _{n}$ whose components are isomorphic to $%
P_{2}\square P_{2}$.
\end{theorem}

\begin{proof}
We will construct the desired $\mathrm{PDDS}$ by applying Corollary \ref{C}.
Set $H=P_{2}\square P_{2}$. We place the graph $D=(V,E)$ that is isomorphic
to $H^{\ast }$ in such a way that $V$ comprises the vertices $O$, $e_{1}$, $%
e_{2}$ and $e_{1}+e_{2}$ and their $24k+8$ neighbors; namely, $-e_{1}$, $%
2e_{1}$, $e_{2}-e_{1}$, $e_{2}+2e_{1}$, $-e_{2}$, $2e_{2}$, $e_{1}-e_{2}$, $%
e_{1}+2e_{2},$ and, if $k>0,$ then also vertices $\pm e_{i}$, $e_{1}\pm
e_{i} $, $e_{2}\pm e_{i}$ and $e_{1}+e_{2}\pm e_{i}$ for $i=3,\ldots ,3k+2$.
Thus, $|V|=24k+12,$ and $D$ contains the vertices $O$ and $e_{i}$, for $%
i=1,\ldots ,n$, as required by Corollary \ref{C}. We set $G=\mathbb{Z}%
_{24k+12}$. The elements $g_{i}$ of $G$ that are assigned to the vertices $%
e_{i}$, for $i=1,\ldots ,n$, are: $g_{1}=2+4k$, $g_{2}=3+6k$, and, if $k>0,$
then $g_{2+i}=2+4k+i$, $g_{2+k+i}=2+4k-i$ and $g_{2+2k+i}=6+11k+i$, for $%
i=1,\ldots ,k$. To finish the proof, we need to show that the restriction of
the mapping $\Phi ((a_{1},\ldots ,a_{n}))=\Phi (e_{1})^{a_{1}}\circ \ldots
\circ \Phi (e_{n})^{a_{n}}=a_{1}g_{1}+\ldots +a_{n}g_{n}$ to the set $V$ is
a bijection. To see this, it suffices to check the table below (broken into two parts to be pasted together horizontally) that shows
that each element of $\mathbb{Z}_{24k+12}$ belongs to the set $\Phi (V).$ In
the table the symbol $[a,b]$ stands for the set $\{a,a+1,a+2,...,b\}.$ In
all cells of the table, the index $i$ runs through the interval $[1,12+24k],$
where $12+24k\equiv 0$ in $G=\mathbb{Z}_{24k+12}$ and integers on the
columns corresponding to $G$ shown in increasing order from left to right,
line by line, and then from top to bottom:
$$
\begin{array}{|l|l|l||l}
\hline
V & \Phi(V) & G &\\\hline
^{e_1-e_{2+k+i}}_{e_2-e_{2+k+i}} &
^i_{1+2k +i} & ^{[1,k]}_{[2+2k,1+3k]} & ^{\ldots}_{\ldots}\\
^{e_{1+k+i}}_{e_1+e_2-e_{2+k+i}} &
^{2+4k+i}_{3+6k+i} & ^{[3+4k,2+
5k]}_{[4+6k,3+7k]} & ^{\ldots}_{\ldots}\\
^{e_1+e_{2+i}}_{e_2+e_{2+i}} & ^{4+8k+
i}_{5+10k+i} & ^{[5+8k,4+9k]}_{[6+10k,5+
11k]} & ^{\ldots}_{\ldots}\\
^{-e_{2+2k+i}}_{e_1+e_2+e_{2+i}} &
^{7+13k-i}_{7+14k+i} & ^{[7+12k,6+
13k]}_{[8+14k,7+15k]} & ^{\ldots}_{\ldots}\\
^{e_2+e_{2+k+i}}_{e_2-e_{2+2k+i}} &
^{9+17k-i}_{10+19k-i} & ^{[9+16k,8+
17k]}_{[10+18k,9+19k]} & ^{\ldots}_{\ldots}\\
^{-e_{2+k+i}}_{e_1+e_2-e_{2+2k+i}} &
^{10+20k+i}_{12+23k-i} & ^{[11+20k,10+
21k]}_{[12+22k,11+23k]}&^{\ldots}_{\ldots}\\
  &  & &\\\hline
\end{array}%
$$
$$
\begin{array}{l||l|l|l||l|l|l|l|}
\hline
&V & \Phi(V) & G & V &\Phi(V) & G \\ \hline
^{\ldots}_{\ldots}&
^{e_2-e_{2+i}}_{e_{2+k+i}} & ^{1+2k-
i}_{2+4k-i} & ^{[k+1,2k]}_{[2+3k,1+4k]}
&  ^{e_2-e_1}_{e_1} & ^{1+2k}_{2+4k} &
^{1+2k}_{2+4k} \\
^{\ldots}_{\ldots}&^{e_1+e_2-e_{2+i}}_{e_1+e_{2+k+i}} &
^{3+6k-i}_{4+8k-i} & ^{[3+5k,2+
6k]}_{[4+7k,3+8k]} &  ^{e_2}_{2e_1} &
^{3+6k}_{4+8k} & ^{3+6k}_{4+8k} \\
^{\ldots}_{\ldots}&^{e_2+e_{2+k+i}}_{e_{2+2k+i}} &
^{5+10k-i}_{5+11k+i} & ^{[5+9k,4+
10k]}_{[6+11k,5+12k]} &  ^{e_1+e_2}_{2e_2} &
^{5+10k}_{6+12k} & ^{5+10k}_{6+12k} \\
^{\ldots}_{\ldots}&^{e_1+e_2+e_{2+k+i}}_{e_1+e_{2+2k+i}} &
^{7+14k-i}_{7+15k+i} & ^{[7+13k,6+
14k]}_{[8+15k,7+16k]} &
^{2e_1+e_2}_{e_1+2e_2} & ^{7+14k}_{8+16k}
& ^{7+14k}_{8+16k} \\
^{\ldots}_{\ldots}&^{e_2+e_{2+2k+i}}_{-e_{2+i}} &
^{8+17k+i}_{10+20k-i} & ^{[9+17k,8+
18k]}_{[10+19k,9+20k]} &  ^{-e_2}_{-e_1} &
^{9+18k}_{10+20k} & ^{9+18k}_{10+20k} \\
^{\ldots}_{\ldots}&^{e_1+e_2+e_{2+2k+i}}_{e_1-e_{2+i}} &
^{10+21k+i}_{12+24k-i} & ^{[11+21k,10+
22k]}_{[12+23k,11+24k]} &  ^{e_1-e_2}_{O} &
^{11+22k}_{12+24k} & ^{11+22k}_{12+24k}
\\
&&  &  &  &  &  \\ \hline
\end{array}%
$$

\noindent As usual at the end of the proof we provide three small examples
for $n=2,5,$ and $8,$ to illustrate it.
$$\begin{array}{||cccc||}\hline
<&e_1,&e_2&>\\\hline
_{10} & ^9_\mathbf0 &^{11}_{\hspace*{1mm}\mathbf2} & _4                     \\
^{\hspace*{2mm}1} & ^\mathbf3_6 & ^{\hspace*{1mm}\mathbf5}_{\hspace*{1mm}8} & ^7\\\hline
\end{array}$$

$$\begin{array}{||cccc||cc|cc||cc|cc||cc|cc||}\hline
<&e_1,&e_2&>&+&e_3&-&e_3&+&e_4&-&e_4&-&e_5&+&e_5\\\hline
_{30}&^{\,27}_{\,\,\,\mathbf0}&^{\,33}_{\,\,\,\mathbf6}&_{12}&_{\,\,\,\mathbf7}&_{13}&_{29}&_{35}&_{\,\,\,\mathbf5}&_{11}&_{31}&_{\,\,\,1}&_\mathbf{17}&_{23}&_{19}&_{25}\\
^{\,\,\,3}&^{\,\,\,\mathbf9}_{\,18}&^\mathbf{15}_{\,24}&^{21}&^{16}&^{22}&^{\,\,\,2}&^{\,\,\,8}&^{14}&^{20}&^{\,\,\,4}&^{10}&^{26}&^{32}&^{28}&^{34}\\\hline
\end{array}$$

$$\begin{array}{||cccc||cc|cc||cc|cc||cc|cc||}\hline
<&e_1,&e_2&>&+&e_3&-&e_3&+&e_5&-&e_5&-&e_7&+&e_7\\\hline
_{50}&^{\,45}_{\,\,\,\mathbf0}&^{\,55}_\mathbf{10}&_{20}&_\mathbf{11}&_{21}&_{49}&_{59}&_{\,\,\,\mathbf9}&_{19}&_{51}&_{\,\,\,1}&_\mathbf{29}&_{39}&_{31}&_{41}\\
^{\,\,\,5}&^\mathbf{15}_{\,30}&^\mathbf{25}_{\,40}&^{35} &^{26}&^{36}&^{\,\,\,4}&^{14}&^{24}&^{34}&^{\,\,\,6}&^{16}&^{44}&^{54}&^{46}&^{56}\\\hline\hline
&&&&+&e_4&-&e_4&+&e_6&-&e_6&-&e_8&+&e_8\\\hline
&&&&&&&&&&&&&&&\\
&&& &^\mathbf{12}_{27}&^{22}_{37}&^{48}_{\,\,\,3}&^{58}_{13}&^{\,\,\,\mathbf8}_{23}&^{18}_{33}&^{52}_{\,\,\,7}&^{\,\,\,2}_{17}&^\mathbf{28}_{43}&^{38}_{53}&^{32}_{47}&^{42}_{57}\\
&&& &&&&&&&&&&&&\\\hline
\end{array}$$
\end{proof}

\section{Proof of Theorem \protect\ref{T1}}

\noindent In this section we prove
%part $\mathrm{\mathbf{(i)}}$ of
Theorem %
\ref{T1}. %Part $\mathrm{\mathbf{(ii)}}$ has been proved in \cite{AD0}.

\begin{proof}
Suppose that there is a $t$-$\mathrm{PDDS}$ $R$ in $\Lambda _{2}$ whose
components are isomorphic to $P_{k}\square P_{s},$ where $k\geq s\geq 3.$
Let $H^{\ast }$ be an induced subgraph of $\Lambda _{n}$ comprising the
vertices of a copy $H$ of $P_{k}\square P_{s}$ and all vertices at distance
at most $t$ from $H$. Clearly $R$ generates a decomposition of $\mathbb{Z}%
^{2}$ into copies of $H^{\ast }$. Although $R$ is not necessarily
lattice-like, all components of $R$ have to be either "parallel"  to the $x$%
-axis, or to be "parallel" to the $y$-axis. Assume wlog that $R$ contains a
component $P_{k}\square P_{s}$ comprising vertices $(x,y),$ where $1\leq
x\leq k,t+1\leq y\leq t+s;$ see the figure below for examples of this
situation for $k=6$, $s=3$ and $t=3$. Consider a set of vertices $%
A=\{(x,0),1\leq x\leq k\}.$ We will show that the vertices of $A$ cannot be
covered by vertex-disjoint copies of $H^{\ast }.$ Assume that a copy of $%
H^{\ast }$ covers only vertices $(x,0),1\leq x\leq m,m<k,$ see the left
example below, where $m=4$. Then the vertex $(m+1,0)$ cannot be covered in $%
R.$ However, if all vertices in $A$ are covered in $R$ by the same copy of $%
H^{\ast }$ (in this case the two copies of $H^{\ast }$ have to be "parallel"
as $k\geq s$) , then the vertices $(k+1,0)$ and $(k+1,1)$ can be covered
only if $s=2,$ a contradiction as we consider the case $s\geq 3.$ See the
right example in the figure.
\end{proof}

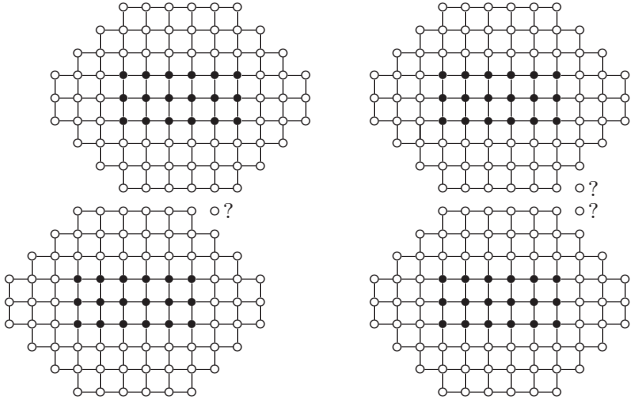
\begin{figure}[htp]
\unitlength=0.60mm \special{em:linewidth 0.4pt} \linethickness{0.4pt}
\begin{picture}(165.00,87.00)
\put(59.00,56.00){\circle{2.00}}
\put(64.00,56.00){\circle{2.00}}
\put(59.00,61.00){\circle*{2.00}}
\put(64.00,61.00){\circle*{2.00}}
\put(69.00,61.00){\circle*{2.00}}
\put(74.00,61.00){\circle*{2.00}}
\put(79.00,61.00){\circle*{2.00}}
\put(59.00,66.00){\circle*{2.00}}
\put(64.00,66.00){\circle*{2.00}}
\put(69.00,66.00){\circle*{2.00}}
\put(74.00,66.00){\circle*{2.00}}
\put(79.00,66.00){\circle*{2.00}}
\put(59.00,71.00){\circle*{2.00}}
\put(64.00,71.00){\circle*{2.00}}
\put(69.00,71.00){\circle*{2.00}}
\put(74.00,71.00){\circle*{2.00}}
\put(79.00,71.00){\circle*{2.00}}
\emline{59.00}{60.00}{1}{59.00}{57.00}{2}
\put(69.00,56.00){\circle{2.00}}
\put(74.00,56.00){\circle{2.00}}
\put(79.00,56.00){\circle{2.00}}
\emline{64.00}{60.00}{3}{64.00}{57.00}{4}
\emline{69.00}{60.00}{5}{69.00}{57.00}{6}
\emline{74.00}{60.00}{7}{74.00}{57.00}{8}
\emline{79.00}{60.00}{9}{79.00}{57.00}{10}
\emline{60.00}{56.00}{11}{63.00}{56.00}{12}
\emline{65.00}{56.00}{13}{68.00}{56.00}{14}
\emline{70.00}{56.00}{15}{73.00}{56.00}{16}
\emline{75.00}{56.00}{17}{78.00}{56.00}{18}
\put(59.00,51.00){\circle{2.00}}
\put(64.00,51.00){\circle{2.00}}
\emline{59.00}{55.00}{19}{59.00}{52.00}{20}
\put(69.00,51.00){\circle{2.00}}
\put(74.00,51.00){\circle{2.00}}
\put(79.00,51.00){\circle{2.00}}
\emline{64.00}{55.00}{21}{64.00}{52.00}{22}
\emline{69.00}{55.00}{23}{69.00}{52.00}{24}
\emline{74.00}{55.00}{25}{74.00}{52.00}{26}
\emline{79.00}{55.00}{27}{79.00}{52.00}{28}
\emline{60.00}{51.00}{29}{63.00}{51.00}{30}
\emline{65.00}{51.00}{31}{68.00}{51.00}{32}
\emline{70.00}{51.00}{33}{73.00}{51.00}{34}
\emline{75.00}{51.00}{35}{78.00}{51.00}{36}
\put(84.00,61.00){\circle{2.00}}
\emline{84.00}{65.00}{37}{84.00}{62.00}{38}
\emline{80.00}{61.00}{39}{83.00}{61.00}{40}
\put(84.00,56.00){\circle{2.00}}
\emline{84.00}{60.00}{41}{84.00}{57.00}{42}
\emline{80.00}{56.00}{43}{83.00}{56.00}{44}
\put(84.00,71.00){\circle{2.00}}
\emline{84.00}{75.00}{45}{84.00}{72.00}{46}
\emline{80.00}{71.00}{47}{83.00}{71.00}{48}
\put(84.00,66.00){\circle{2.00}}
\emline{84.00}{70.00}{49}{84.00}{67.00}{50}
\emline{80.00}{66.00}{51}{83.00}{66.00}{52}
\put(89.00,61.00){\circle{2.00}}
\emline{89.00}{65.00}{53}{89.00}{62.00}{54}
\emline{85.00}{61.00}{55}{88.00}{61.00}{56}
\put(89.00,71.00){\circle{2.00}}
\emline{85.00}{71.00}{57}{88.00}{71.00}{58}
\put(89.00,66.00){\circle{2.00}}
\emline{89.00}{70.00}{59}{89.00}{67.00}{60}
\emline{85.00}{66.00}{61}{88.00}{66.00}{62}
\put(59.00,76.00){\circle{2.00}}
\put(64.00,76.00){\circle{2.00}}
\emline{59.00}{80.00}{63}{59.00}{77.00}{64}
\put(69.00,76.00){\circle{2.00}}
\put(74.00,76.00){\circle{2.00}}
\put(79.00,76.00){\circle{2.00}}
\emline{64.00}{80.00}{65}{64.00}{77.00}{66}
\emline{69.00}{80.00}{67}{69.00}{77.00}{68}
\emline{74.00}{80.00}{69}{74.00}{77.00}{70}
\emline{79.00}{80.00}{71}{79.00}{77.00}{72}
\emline{60.00}{76.00}{73}{63.00}{76.00}{74}
\emline{65.00}{76.00}{75}{68.00}{76.00}{76}
\emline{70.00}{76.00}{77}{73.00}{76.00}{78}
\emline{75.00}{76.00}{79}{78.00}{76.00}{80}
\emline{59.00}{75.00}{81}{59.00}{72.00}{82}
\emline{64.00}{75.00}{83}{64.00}{72.00}{84}
\emline{69.00}{75.00}{85}{69.00}{72.00}{86}
\emline{74.00}{75.00}{87}{74.00}{72.00}{88}
\emline{79.00}{75.00}{89}{79.00}{72.00}{90}
\put(84.00,76.00){\circle{2.00}}
\emline{80.00}{76.00}{91}{83.00}{76.00}{92}
\put(59.00,81.00){\circle{2.00}}
\put(64.00,81.00){\circle{2.00}}
\put(69.00,81.00){\circle{2.00}}
\put(74.00,81.00){\circle{2.00}}
\put(79.00,81.00){\circle{2.00}}
\emline{60.00}{81.00}{93}{63.00}{81.00}{94}
\emline{65.00}{81.00}{95}{68.00}{81.00}{96}
\emline{70.00}{81.00}{97}{73.00}{81.00}{98}
\emline{75.00}{81.00}{99}{78.00}{81.00}{100}
\put(54.00,56.00){\circle{2.00}}
\emline{54.00}{60.00}{101}{54.00}{57.00}{102}
\emline{54.00}{75.00}{103}{54.00}{72.00}{104}
\put(54.00,76.00){\circle{2.00}}
\emline{55.00}{56.00}{105}{58.00}{56.00}{106}
\emline{55.00}{76.00}{107}{58.00}{76.00}{108}
\put(49.00,61.00){\circle{2.00}}
\emline{49.00}{65.00}{109}{49.00}{62.00}{110}
\put(49.00,71.00){\circle{2.00}}
\put(49.00,66.00){\circle{2.00}}
\emline{49.00}{70.00}{111}{49.00}{67.00}{112}
\emline{50.00}{61.00}{113}{53.00}{61.00}{114}
\emline{50.00}{71.00}{115}{53.00}{71.00}{116}
\emline{50.00}{66.00}{117}{53.00}{66.00}{118}
\put(49.00,11.00){\circle{2.00}}
\put(54.00,11.00){\circle{2.00}}
\put(49.00,16.00){\circle*{2.00}}
\put(54.00,16.00){\circle*{2.00}}
\put(59.00,16.00){\circle*{2.00}}
\put(64.00,16.00){\circle*{2.00}}
\put(69.00,16.00){\circle*{2.00}}
\put(49.00,21.00){\circle*{2.00}}
\put(54.00,21.00){\circle*{2.00}}
\put(59.00,21.00){\circle*{2.00}}
\put(64.00,21.00){\circle*{2.00}}
\put(69.00,21.00){\circle*{2.00}}
\put(49.00,26.00){\circle*{2.00}}
\put(54.00,26.00){\circle*{2.00}}
\put(59.00,26.00){\circle*{2.00}}
\put(64.00,26.00){\circle*{2.00}}
\put(69.00,26.00){\circle*{2.00}}
\emline{49.00}{15.00}{119}{49.00}{12.00}{120}
\put(59.00,11.00){\circle{2.00}}
\put(64.00,11.00){\circle{2.00}}
\put(69.00,11.00){\circle{2.00}}
\emline{54.00}{15.00}{121}{54.00}{12.00}{122}
\emline{59.00}{15.00}{123}{59.00}{12.00}{124}
\emline{64.00}{15.00}{125}{64.00}{12.00}{126}
\emline{69.00}{15.00}{127}{69.00}{12.00}{128}
\emline{50.00}{11.00}{129}{53.00}{11.00}{130}
\emline{55.00}{11.00}{131}{58.00}{11.00}{132}
\emline{60.00}{11.00}{133}{63.00}{11.00}{134}
\emline{65.00}{11.00}{135}{68.00}{11.00}{136}
\put(49.00,6.00){\circle{2.00}}
\put(54.00,6.00){\circle{2.00}}
\emline{49.00}{10.00}{137}{49.00}{7.00}{138}
\put(59.00,6.00){\circle{2.00}}
\put(64.00,6.00){\circle{2.00}}
\put(69.00,6.00){\circle{2.00}}
\emline{54.00}{10.00}{139}{54.00}{7.00}{140}
\emline{59.00}{10.00}{141}{59.00}{7.00}{142}
\emline{64.00}{10.00}{143}{64.00}{7.00}{144}
\emline{69.00}{10.00}{145}{69.00}{7.00}{146}
\emline{50.00}{6.00}{147}{53.00}{6.00}{148}
\emline{55.00}{6.00}{149}{58.00}{6.00}{150}
\emline{60.00}{6.00}{151}{63.00}{6.00}{152}
\emline{65.00}{6.00}{153}{68.00}{6.00}{154}
\put(74.00,16.00){\circle{2.00}}
\emline{74.00}{20.00}{155}{74.00}{17.00}{156}
\emline{70.00}{16.00}{157}{73.00}{16.00}{158}
\put(74.00,11.00){\circle{2.00}}
\emline{74.00}{15.00}{159}{74.00}{12.00}{160}
\emline{70.00}{11.00}{161}{73.00}{11.00}{162}
\put(74.00,26.00){\circle{2.00}}
\emline{74.00}{35.00}{163}{74.00}{32.00}{164}
\emline{70.00}{26.00}{165}{73.00}{26.00}{166}
\put(74.00,21.00){\circle{2.00}}
\emline{74.00}{25.00}{167}{74.00}{22.00}{168}
\emline{70.00}{21.00}{169}{73.00}{21.00}{170}
\put(79.00,16.00){\circle{2.00}}
\emline{79.00}{20.00}{171}{79.00}{17.00}{172}
\emline{75.00}{16.00}{173}{78.00}{16.00}{174}
\put(79.00,31.00){\circle{2.00}}
\emline{75.00}{31.00}{175}{78.00}{31.00}{176}
\put(79.00,21.00){\circle{2.00}}
\emline{79.00}{30.00}{177}{79.00}{27.00}{178}
\emline{75.00}{21.00}{179}{78.00}{21.00}{180}
\put(49.00,31.00){\circle{2.00}}
\put(54.00,31.00){\circle{2.00}}
\emline{49.00}{40.00}{181}{49.00}{37.00}{182}
\put(59.00,31.00){\circle{2.00}}
\put(64.00,31.00){\circle{2.00}}
\put(69.00,31.00){\circle{2.00}}
\emline{54.00}{40.00}{183}{54.00}{37.00}{184}
\emline{59.00}{40.00}{185}{59.00}{37.00}{186}
\emline{64.00}{40.00}{187}{64.00}{37.00}{188}
\emline{69.00}{40.00}{189}{69.00}{37.00}{190}
\emline{50.00}{31.00}{191}{53.00}{31.00}{192}
\emline{55.00}{31.00}{193}{58.00}{31.00}{194}
\emline{60.00}{31.00}{195}{63.00}{31.00}{196}
\emline{65.00}{31.00}{197}{68.00}{31.00}{198}
\emline{49.00}{30.00}{199}{49.00}{27.00}{200}
\emline{54.00}{30.00}{201}{54.00}{27.00}{202}
\emline{59.00}{30.00}{203}{59.00}{27.00}{204}
\emline{64.00}{30.00}{205}{64.00}{27.00}{206}
\emline{69.00}{30.00}{207}{69.00}{27.00}{208}
\emline{70.00}{36.00}{209}{73.00}{36.00}{210}
\put(49.00,41.00){\circle{2.00}}
\put(54.00,41.00){\circle{2.00}}
\put(59.00,41.00){\circle{2.00}}
\put(64.00,41.00){\circle{2.00}}
\put(69.00,41.00){\circle{2.00}}
\emline{50.00}{41.00}{211}{53.00}{41.00}{212}
\emline{55.00}{41.00}{213}{58.00}{41.00}{214}
\emline{60.00}{41.00}{215}{63.00}{41.00}{216}
\emline{65.00}{41.00}{217}{68.00}{41.00}{218}
\put(44.00,11.00){\circle{2.00}}
\emline{44.00}{15.00}{219}{44.00}{12.00}{220}
\emline{44.00}{30.00}{221}{44.00}{27.00}{222}
\put(44.00,31.00){\circle{2.00}}
\emline{45.00}{11.00}{223}{48.00}{11.00}{224}
\emline{45.00}{31.00}{225}{48.00}{31.00}{226}
\put(39.00,16.00){\circle{2.00}}
\emline{39.00}{20.00}{227}{39.00}{17.00}{228}
\put(39.00,26.00){\circle{2.00}}
\put(39.00,21.00){\circle{2.00}}
\emline{39.00}{25.00}{229}{39.00}{22.00}{230}
\emline{40.00}{16.00}{231}{43.00}{16.00}{232}
\emline{40.00}{26.00}{233}{43.00}{26.00}{234}
\emline{40.00}{21.00}{235}{43.00}{21.00}{236}
\put(74.00,41.00){\circle{2.00}}
\put(77.00,41.00){\makebox(0,0)[cc]{\small{?}}}
\put(129.00,56.00){\circle{2.00}}
\put(134.00,56.00){\circle{2.00}}
\put(129.00,61.00){\circle*{2.00}}
\put(134.00,61.00){\circle*{2.00}}
\put(139.00,61.00){\circle*{2.00}}
\put(144.00,61.00){\circle*{2.00}}
\put(149.00,61.00){\circle*{2.00}}
\put(129.00,66.00){\circle*{2.00}}
\put(134.00,66.00){\circle*{2.00}}
\put(139.00,66.00){\circle*{2.00}}
\put(144.00,66.00){\circle*{2.00}}
\put(149.00,66.00){\circle*{2.00}}
\put(129.00,71.00){\circle*{2.00}}
\put(134.00,71.00){\circle*{2.00}}
\put(139.00,71.00){\circle*{2.00}}
\put(144.00,71.00){\circle*{2.00}}
\put(149.00,71.00){\circle*{2.00}}
\emline{129.00}{60.00}{237}{129.00}{57.00}{238}
\put(139.00,56.00){\circle{2.00}}
\put(144.00,56.00){\circle{2.00}}
\put(149.00,56.00){\circle{2.00}}
\emline{134.00}{60.00}{239}{134.00}{57.00}{240}
\emline{139.00}{60.00}{241}{139.00}{57.00}{242}
\emline{144.00}{60.00}{243}{144.00}{57.00}{244}
\emline{149.00}{60.00}{245}{149.00}{57.00}{246}
\emline{130.00}{56.00}{247}{133.00}{56.00}{248}
\emline{135.00}{56.00}{249}{138.00}{56.00}{250}
\emline{140.00}{56.00}{251}{143.00}{56.00}{252}
\emline{145.00}{56.00}{253}{148.00}{56.00}{254}
\put(129.00,51.00){\circle{2.00}}
\put(134.00,51.00){\circle{2.00}}
\emline{129.00}{55.00}{255}{129.00}{52.00}{256}
\put(139.00,51.00){\circle{2.00}}
\put(144.00,51.00){\circle{2.00}}
\put(149.00,51.00){\circle{2.00}}
\emline{134.00}{55.00}{257}{134.00}{52.00}{258}
\emline{139.00}{55.00}{259}{139.00}{52.00}{260}
\emline{144.00}{55.00}{261}{144.00}{52.00}{262}
\emline{149.00}{55.00}{263}{149.00}{52.00}{264}
\emline{130.00}{51.00}{265}{133.00}{51.00}{266}
\emline{135.00}{51.00}{267}{138.00}{51.00}{268}
\emline{140.00}{51.00}{269}{143.00}{51.00}{270}
\emline{145.00}{51.00}{271}{148.00}{51.00}{272}
\put(154.00,61.00){\circle{2.00}}
\emline{154.00}{65.00}{273}{154.00}{62.00}{274}
\emline{150.00}{61.00}{275}{153.00}{61.00}{276}
\put(154.00,56.00){\circle{2.00}}
\emline{154.00}{60.00}{277}{154.00}{57.00}{278}
\emline{150.00}{56.00}{279}{153.00}{56.00}{280}
\put(154.00,71.00){\circle{2.00}}
\emline{154.00}{75.00}{281}{154.00}{72.00}{282}
\emline{150.00}{71.00}{283}{153.00}{71.00}{284}
\put(154.00,66.00){\circle{2.00}}
\emline{154.00}{70.00}{285}{154.00}{67.00}{286}
\emline{150.00}{66.00}{287}{153.00}{66.00}{288}
\put(159.00,61.00){\circle{2.00}}
\emline{159.00}{65.00}{289}{159.00}{62.00}{290}
\emline{155.00}{61.00}{291}{158.00}{61.00}{292}
\put(159.00,71.00){\circle{2.00}}
\emline{155.00}{71.00}{293}{158.00}{71.00}{294}
\put(159.00,66.00){\circle{2.00}}
\emline{159.00}{70.00}{295}{159.00}{67.00}{296}
\emline{155.00}{66.00}{297}{158.00}{66.00}{298}
\put(129.00,76.00){\circle{2.00}}
\put(134.00,76.00){\circle{2.00}}
\emline{129.00}{80.00}{299}{129.00}{77.00}{300}
\put(139.00,76.00){\circle{2.00}}
\put(144.00,76.00){\circle{2.00}}
\put(149.00,76.00){\circle{2.00}}
\emline{134.00}{80.00}{301}{134.00}{77.00}{302}
\emline{139.00}{80.00}{303}{139.00}{77.00}{304}
\emline{144.00}{80.00}{305}{144.00}{77.00}{306}
\emline{149.00}{80.00}{307}{149.00}{77.00}{308}
\emline{130.00}{76.00}{309}{133.00}{76.00}{310}
\emline{135.00}{76.00}{311}{138.00}{76.00}{312}
\emline{140.00}{76.00}{313}{143.00}{76.00}{314}
\emline{145.00}{76.00}{315}{148.00}{76.00}{316}
\emline{129.00}{75.00}{317}{129.00}{72.00}{318}
\emline{134.00}{75.00}{319}{134.00}{72.00}{320}
\emline{139.00}{75.00}{321}{139.00}{72.00}{322}
\emline{144.00}{75.00}{323}{144.00}{72.00}{324}
\emline{149.00}{75.00}{325}{149.00}{72.00}{326}
\put(154.00,76.00){\circle{2.00}}
\emline{150.00}{76.00}{327}{153.00}{76.00}{328}
\put(129.00,81.00){\circle{2.00}}
\put(134.00,81.00){\circle{2.00}}
\put(139.00,81.00){\circle{2.00}}
\put(144.00,81.00){\circle{2.00}}
\put(149.00,81.00){\circle{2.00}}
\emline{130.00}{81.00}{329}{133.00}{81.00}{330}
\emline{135.00}{81.00}{331}{138.00}{81.00}{332}
\emline{140.00}{81.00}{333}{143.00}{81.00}{334}
\emline{145.00}{81.00}{335}{148.00}{81.00}{336}
\put(124.00,56.00){\circle{2.00}}
\emline{124.00}{60.00}{337}{124.00}{57.00}{338}
\emline{124.00}{75.00}{339}{124.00}{72.00}{340}
\put(124.00,76.00){\circle{2.00}}
\emline{125.00}{56.00}{341}{128.00}{56.00}{342}
\emline{125.00}{76.00}{343}{128.00}{76.00}{344}
\put(119.00,61.00){\circle{2.00}}
\emline{119.00}{65.00}{345}{119.00}{62.00}{346}
\put(119.00,71.00){\circle{2.00}}
\put(119.00,66.00){\circle{2.00}}
\emline{119.00}{70.00}{347}{119.00}{67.00}{348}
\emline{120.00}{61.00}{349}{123.00}{61.00}{350}
\emline{120.00}{71.00}{351}{123.00}{71.00}{352}
\emline{120.00}{66.00}{353}{123.00}{66.00}{354}
\put(129.00,11.00){\circle{2.00}}
\put(134.00,11.00){\circle{2.00}}
\put(129.00,16.00){\circle*{2.00}}
\put(134.00,16.00){\circle*{2.00}}
\put(139.00,16.00){\circle*{2.00}}
\put(144.00,16.00){\circle*{2.00}}
\put(149.00,16.00){\circle*{2.00}}
\put(129.00,21.00){\circle*{2.00}}
\put(134.00,21.00){\circle*{2.00}}
\put(139.00,21.00){\circle*{2.00}}
\put(144.00,21.00){\circle*{2.00}}
\put(149.00,21.00){\circle*{2.00}}
\put(129.00,26.00){\circle*{2.00}}
\put(134.00,26.00){\circle*{2.00}}
\put(139.00,26.00){\circle*{2.00}}
\put(144.00,26.00){\circle*{2.00}}
\put(149.00,26.00){\circle*{2.00}}
\emline{129.00}{15.00}{355}{129.00}{12.00}{356}
\put(139.00,11.00){\circle{2.00}}
\put(144.00,11.00){\circle{2.00}}
\put(149.00,11.00){\circle{2.00}}
\emline{134.00}{15.00}{357}{134.00}{12.00}{358}
\emline{139.00}{15.00}{359}{139.00}{12.00}{360}
\emline{144.00}{15.00}{361}{144.00}{12.00}{362}
\emline{149.00}{15.00}{363}{149.00}{12.00}{364}
\emline{130.00}{11.00}{365}{133.00}{11.00}{366}
\emline{135.00}{11.00}{367}{138.00}{11.00}{368}
\emline{140.00}{11.00}{369}{143.00}{11.00}{370}
\emline{145.00}{11.00}{371}{148.00}{11.00}{372}
\put(129.00,6.00){\circle{2.00}}
\put(134.00,6.00){\circle{2.00}}
\emline{129.00}{10.00}{373}{129.00}{7.00}{374}
\put(139.00,6.00){\circle{2.00}}
\put(144.00,6.00){\circle{2.00}}
\put(149.00,6.00){\circle{2.00}}
\emline{134.00}{10.00}{375}{134.00}{7.00}{376}
\emline{139.00}{10.00}{377}{139.00}{7.00}{378}
\emline{144.00}{10.00}{379}{144.00}{7.00}{380}
\emline{149.00}{10.00}{381}{149.00}{7.00}{382}
\emline{130.00}{6.00}{383}{133.00}{6.00}{384}
\emline{135.00}{6.00}{385}{138.00}{6.00}{386}
\emline{140.00}{6.00}{387}{143.00}{6.00}{388}
\emline{145.00}{6.00}{389}{148.00}{6.00}{390}
\put(154.00,16.00){\circle{2.00}}
\emline{154.00}{20.00}{391}{154.00}{17.00}{392}
\emline{150.00}{16.00}{393}{153.00}{16.00}{394}
\put(154.00,11.00){\circle{2.00}}
\emline{154.00}{15.00}{395}{154.00}{12.00}{396}
\emline{150.00}{11.00}{397}{153.00}{11.00}{398}
\put(154.00,26.00){\circle{2.00}}
\emline{154.00}{35.00}{399}{154.00}{32.00}{400}
\emline{150.00}{26.00}{401}{153.00}{26.00}{402}
\put(154.00,21.00){\circle{2.00}}
\emline{154.00}{25.00}{403}{154.00}{22.00}{404}
\emline{150.00}{21.00}{405}{153.00}{21.00}{406}
\put(159.00,16.00){\circle{2.00}}
\emline{159.00}{20.00}{407}{159.00}{17.00}{408}
\emline{155.00}{16.00}{409}{158.00}{16.00}{410}
\put(159.00,31.00){\circle{2.00}}
\emline{155.00}{31.00}{411}{158.00}{31.00}{412}
\put(159.00,21.00){\circle{2.00}}
\emline{159.00}{30.00}{413}{159.00}{27.00}{414}
\emline{155.00}{21.00}{415}{158.00}{21.00}{416}
\put(129.00,31.00){\circle{2.00}}
\put(134.00,31.00){\circle{2.00}}
\emline{129.00}{40.00}{417}{129.00}{37.00}{418}
\put(139.00,31.00){\circle{2.00}}
\put(144.00,31.00){\circle{2.00}}
\put(149.00,31.00){\circle{2.00}}
\emline{134.00}{40.00}{419}{134.00}{37.00}{420}
\emline{139.00}{40.00}{421}{139.00}{37.00}{422}
\emline{144.00}{40.00}{423}{144.00}{37.00}{424}
\emline{149.00}{40.00}{425}{149.00}{37.00}{426}
\emline{130.00}{31.00}{427}{133.00}{31.00}{428}
\emline{135.00}{31.00}{429}{138.00}{31.00}{430}
\emline{140.00}{31.00}{431}{143.00}{31.00}{432}
\emline{145.00}{31.00}{433}{148.00}{31.00}{434}
\emline{129.00}{30.00}{435}{129.00}{27.00}{436}
\emline{134.00}{30.00}{437}{134.00}{27.00}{438}
\emline{139.00}{30.00}{439}{139.00}{27.00}{440}
\emline{144.00}{30.00}{441}{144.00}{27.00}{442}
\emline{149.00}{30.00}{443}{149.00}{27.00}{444}
\put(154.00,36.00){\circle{2.00}}
\emline{150.00}{36.00}{445}{153.00}{36.00}{446}
\put(129.00,41.00){\circle{2.00}}
\put(134.00,41.00){\circle{2.00}}
\put(139.00,41.00){\circle{2.00}}
\put(144.00,41.00){\circle{2.00}}
\put(149.00,41.00){\circle{2.00}}
\emline{130.00}{41.00}{447}{133.00}{41.00}{448}
\emline{135.00}{41.00}{449}{138.00}{41.00}{450}
\emline{140.00}{41.00}{451}{143.00}{41.00}{452}
\emline{145.00}{41.00}{453}{148.00}{41.00}{454}
\put(124.00,11.00){\circle{2.00}}
\emline{124.00}{15.00}{455}{124.00}{12.00}{456}
\emline{124.00}{30.00}{457}{124.00}{27.00}{458}
\put(124.00,31.00){\circle{2.00}}
\emline{125.00}{11.00}{459}{128.00}{11.00}{460}
\emline{125.00}{31.00}{461}{128.00}{31.00}{462}
\put(119.00,16.00){\circle{2.00}}
\emline{119.00}{20.00}{463}{119.00}{17.00}{464}
\put(119.00,26.00){\circle{2.00}}
\put(119.00,21.00){\circle{2.00}}
\emline{119.00}{25.00}{465}{119.00}{22.00}{466}
\emline{120.00}{16.00}{467}{123.00}{16.00}{468}
\emline{120.00}{26.00}{469}{123.00}{26.00}{470}
\emline{120.00}{21.00}{471}{123.00}{21.00}{472}
\put(154.00,41.00){\circle{2.00}}
\put(157.00,41.00){\makebox(0,0)[cc]{\small{?}}}
\put(154.00,51.00){\circle{2.00}}
\put(44.00,6.00){\circle{2.00}}
\emline{45.00}{6.00}{473}{48.00}{6.00}{474}
\emline{44.00}{7.00}{475}{44.00}{10.00}{476}
\emline{43.00}{11.00}{477}{40.00}{11.00}{478}
\put(39.00,11.00){\circle{2.00}}
\emline{39.00}{12.00}{479}{39.00}{15.00}{480}
\put(44.00,16.00){\circle*{2.00}}
\put(44.00,21.00){\circle*{2.00}}
\put(44.00,26.00){\circle*{2.00}}
\put(54.00,61.00){\circle*{2.00}}
\put(54.00,66.00){\circle*{2.00}}
\put(54.00,71.00){\circle*{2.00}}
\put(54.00,51.00){\circle{2.00}}
\emline{54.00}{55.00}{481}{54.00}{52.00}{482}
\emline{55.00}{51.00}{483}{58.00}{51.00}{484}
\emline{54.00}{80.00}{485}{54.00}{77.00}{486}
\put(54.00,81.00){\circle{2.00}}
\emline{55.00}{81.00}{487}{58.00}{81.00}{488}
\put(49.00,56.00){\circle{2.00}}
\emline{49.00}{60.00}{489}{49.00}{57.00}{490}
\emline{49.00}{75.00}{491}{49.00}{72.00}{492}
\put(49.00,76.00){\circle{2.00}}
\emline{50.00}{56.00}{493}{53.00}{56.00}{494}
\emline{50.00}{76.00}{495}{53.00}{76.00}{496}
\put(44.00,61.00){\circle{2.00}}
\emline{44.00}{65.00}{497}{44.00}{62.00}{498}
\put(44.00,71.00){\circle{2.00}}
\put(44.00,66.00){\circle{2.00}}
\emline{44.00}{70.00}{499}{44.00}{67.00}{500}
\emline{45.00}{61.00}{501}{48.00}{61.00}{502}
\emline{45.00}{71.00}{503}{48.00}{71.00}{504}
\emline{45.00}{66.00}{505}{48.00}{66.00}{506}
\emline{44.00}{40.00}{507}{44.00}{37.00}{508}
\put(44.00,41.00){\circle{2.00}}
\emline{45.00}{41.00}{509}{48.00}{41.00}{510}
\emline{39.00}{35.00}{511}{39.00}{32.00}{512}
\put(39.00,36.00){\circle{2.00}}
\emline{40.00}{36.00}{513}{43.00}{36.00}{514}
\put(34.00,16.00){\circle{2.00}}
\emline{34.00}{20.00}{515}{34.00}{17.00}{516}
\put(34.00,31.00){\circle{2.00}}
\put(34.00,21.00){\circle{2.00}}
\emline{34.00}{30.00}{517}{34.00}{27.00}{518}
\emline{35.00}{16.00}{519}{38.00}{16.00}{520}
\emline{35.00}{31.00}{521}{38.00}{31.00}{522}
\emline{35.00}{21.00}{523}{38.00}{21.00}{524}
\put(124.00,51.00){\circle{2.00}}
\emline{124.00}{55.00}{525}{124.00}{52.00}{526}
\emline{125.00}{51.00}{527}{128.00}{51.00}{528}
\emline{124.00}{80.00}{529}{124.00}{77.00}{530}
\put(124.00,81.00){\circle{2.00}}
\emline{125.00}{81.00}{531}{128.00}{81.00}{532}
\put(119.00,56.00){\circle{2.00}}
\emline{119.00}{60.00}{533}{119.00}{57.00}{534}
\emline{119.00}{75.00}{535}{119.00}{72.00}{536}
\put(119.00,76.00){\circle{2.00}}
\emline{120.00}{56.00}{537}{123.00}{56.00}{538}
\emline{120.00}{76.00}{539}{123.00}{76.00}{540}
\put(114.00,61.00){\circle{2.00}}
\emline{114.00}{65.00}{541}{114.00}{62.00}{542}
\put(114.00,71.00){\circle{2.00}}
\put(114.00,66.00){\circle{2.00}}
\emline{114.00}{70.00}{543}{114.00}{67.00}{544}
\emline{115.00}{61.00}{545}{118.00}{61.00}{546}
\emline{115.00}{71.00}{547}{118.00}{71.00}{548}
\emline{115.00}{66.00}{549}{118.00}{66.00}{550}
\put(124.00,16.00){\circle*{2.00}}
\put(124.00,21.00){\circle*{2.00}}
\put(124.00,26.00){\circle*{2.00}}
\put(124.00,6.00){\circle{2.00}}
\emline{124.00}{10.00}{551}{124.00}{7.00}{552}
\emline{125.00}{6.00}{553}{128.00}{6.00}{554}
\emline{124.00}{40.00}{555}{124.00}{37.00}{556}
\put(124.00,41.00){\circle{2.00}}
\emline{125.00}{41.00}{557}{128.00}{41.00}{558}
\put(119.00,11.00){\circle{2.00}}
\emline{119.00}{15.00}{559}{119.00}{12.00}{560}
\emline{119.00}{35.00}{561}{119.00}{32.00}{562}
\put(119.00,36.00){\circle{2.00}}
\emline{120.00}{11.00}{563}{123.00}{11.00}{564}
\emline{120.00}{36.00}{565}{123.00}{36.00}{566}
\put(114.00,16.00){\circle{2.00}}
\emline{114.00}{20.00}{567}{114.00}{17.00}{568}
\put(114.00,31.00){\circle{2.00}}
\put(114.00,21.00){\circle{2.00}}
\emline{114.00}{30.00}{569}{114.00}{27.00}{570}
\emline{115.00}{16.00}{571}{118.00}{16.00}{572}
\emline{115.00}{31.00}{573}{118.00}{31.00}{574}
\emline{115.00}{21.00}{575}{118.00}{21.00}{576}
\put(124.00,61.00){\circle*{2.00}}
\put(124.00,66.00){\circle*{2.00}}
\put(124.00,71.00){\circle*{2.00}}
\put(49.00,1.00){\circle{2.00}}
\put(54.00,1.00){\circle{2.00}}
\emline{49.00}{5.00}{577}{49.00}{2.00}{578}
\put(59.00,1.00){\circle{2.00}}
\put(64.00,1.00){\circle{2.00}}
\put(69.00,1.00){\circle{2.00}}
\emline{54.00}{5.00}{579}{54.00}{2.00}{580}
\emline{59.00}{5.00}{581}{59.00}{2.00}{582}
\emline{64.00}{5.00}{583}{64.00}{2.00}{584}
\emline{69.00}{5.00}{585}{69.00}{2.00}{586}
\emline{50.00}{1.00}{587}{53.00}{1.00}{588}
\emline{55.00}{1.00}{589}{58.00}{1.00}{590}
\emline{60.00}{1.00}{591}{63.00}{1.00}{592}
\emline{65.00}{1.00}{593}{68.00}{1.00}{594}
\put(74.00,6.00){\circle{2.00}}
\emline{74.00}{10.00}{595}{74.00}{7.00}{596}
\emline{70.00}{6.00}{597}{73.00}{6.00}{598}
\put(79.00,11.00){\circle{2.00}}
\emline{79.00}{15.00}{599}{79.00}{12.00}{600}
\emline{75.00}{11.00}{601}{78.00}{11.00}{602}
\put(129.00,1.00){\circle{2.00}}
\put(134.00,1.00){\circle{2.00}}
\emline{129.00}{5.00}{603}{129.00}{2.00}{604}
\put(139.00,1.00){\circle{2.00}}
\put(144.00,1.00){\circle{2.00}}
\put(149.00,1.00){\circle{2.00}}
\emline{134.00}{5.00}{605}{134.00}{2.00}{606}
\emline{139.00}{5.00}{607}{139.00}{2.00}{608}
\emline{144.00}{5.00}{609}{144.00}{2.00}{610}
\emline{149.00}{5.00}{611}{149.00}{2.00}{612}
\emline{130.00}{1.00}{613}{133.00}{1.00}{614}
\emline{135.00}{1.00}{615}{138.00}{1.00}{616}
\emline{140.00}{1.00}{617}{143.00}{1.00}{618}
\emline{145.00}{1.00}{619}{148.00}{1.00}{620}
\put(154.00,6.00){\circle{2.00}}
\emline{154.00}{10.00}{621}{154.00}{7.00}{622}
\emline{150.00}{6.00}{623}{153.00}{6.00}{624}
\put(159.00,11.00){\circle{2.00}}
\emline{159.00}{15.00}{625}{159.00}{12.00}{626}
\emline{155.00}{11.00}{627}{158.00}{11.00}{628}
\put(44.00,1.00){\circle{2.00}}
\emline{45.00}{1.00}{629}{48.00}{1.00}{630}
\emline{44.00}{2.00}{631}{44.00}{5.00}{632}
\emline{43.00}{6.00}{633}{40.00}{6.00}{634}
\put(39.00,6.00){\circle{2.00}}
\emline{39.00}{7.00}{635}{39.00}{10.00}{636}
\put(34.00,11.00){\circle{2.00}}
\emline{34.00}{15.00}{637}{34.00}{12.00}{638}
\emline{35.00}{11.00}{639}{38.00}{11.00}{640}
\put(124.00,1.00){\circle{2.00}}
\emline{124.00}{5.00}{641}{124.00}{2.00}{642}
\emline{125.00}{1.00}{643}{128.00}{1.00}{644}
\put(119.00,6.00){\circle{2.00}}
\emline{119.00}{10.00}{645}{119.00}{7.00}{646}
\emline{120.00}{6.00}{647}{123.00}{6.00}{648}
\put(114.00,11.00){\circle{2.00}}
\emline{114.00}{15.00}{649}{114.00}{12.00}{650}
\emline{115.00}{11.00}{651}{118.00}{11.00}{652}
\emline{84.00}{80.00}{653}{84.00}{77.00}{654}
\put(89.00,76.00){\circle{2.00}}
\emline{85.00}{76.00}{655}{88.00}{76.00}{656}
\emline{89.00}{75.00}{657}{89.00}{72.00}{658}
\emline{59.00}{85.00}{659}{59.00}{82.00}{660}
\emline{64.00}{85.00}{661}{64.00}{82.00}{662}
\emline{69.00}{85.00}{663}{69.00}{82.00}{664}
\emline{74.00}{85.00}{665}{74.00}{82.00}{666}
\emline{79.00}{85.00}{667}{79.00}{82.00}{668}
\put(84.00,81.00){\circle{2.00}}
\emline{80.00}{81.00}{669}{83.00}{81.00}{670}
\put(59.00,86.00){\circle{2.00}}
\put(64.00,86.00){\circle{2.00}}
\put(69.00,86.00){\circle{2.00}}
\put(74.00,86.00){\circle{2.00}}
\put(79.00,86.00){\circle{2.00}}
\emline{60.00}{86.00}{671}{63.00}{86.00}{672}
\emline{65.00}{86.00}{673}{68.00}{86.00}{674}
\emline{70.00}{86.00}{675}{73.00}{86.00}{676}
\emline{75.00}{86.00}{677}{78.00}{86.00}{678}
\emline{154.00}{80.00}{679}{154.00}{77.00}{680}
\put(159.00,76.00){\circle{2.00}}
\emline{155.00}{76.00}{681}{158.00}{76.00}{682}
\emline{159.00}{75.00}{683}{159.00}{72.00}{684}
\emline{129.00}{85.00}{685}{129.00}{82.00}{686}
\emline{134.00}{85.00}{687}{134.00}{82.00}{688}
\emline{139.00}{85.00}{689}{139.00}{82.00}{690}
\emline{144.00}{85.00}{691}{144.00}{82.00}{692}
\emline{149.00}{85.00}{693}{149.00}{82.00}{694}
\put(154.00,81.00){\circle{2.00}}
\emline{150.00}{81.00}{695}{153.00}{81.00}{696}
\put(129.00,86.00){\circle{2.00}}
\put(134.00,86.00){\circle{2.00}}
\put(139.00,86.00){\circle{2.00}}
\put(144.00,86.00){\circle{2.00}}
\put(149.00,86.00){\circle{2.00}}
\emline{130.00}{86.00}{697}{133.00}{86.00}{698}
\emline{135.00}{86.00}{699}{138.00}{86.00}{700}
\emline{140.00}{86.00}{701}{143.00}{86.00}{702}
\emline{145.00}{86.00}{703}{148.00}{86.00}{704}
\emline{54.00}{85.00}{705}{54.00}{82.00}{706}
\put(54.00,86.00){\circle{2.00}}
\emline{55.00}{86.00}{707}{58.00}{86.00}{708}
\emline{49.00}{80.00}{709}{49.00}{77.00}{710}
\put(49.00,81.00){\circle{2.00}}
\emline{50.00}{81.00}{711}{53.00}{81.00}{712}
\put(44.00,76.00){\circle{2.00}}
\emline{44.00}{75.00}{713}{44.00}{72.00}{714}
\emline{45.00}{76.00}{715}{48.00}{76.00}{716}
\emline{124.00}{85.00}{717}{124.00}{82.00}{718}
\put(124.00,86.00){\circle{2.00}}
\emline{125.00}{86.00}{719}{128.00}{86.00}{720}
\emline{119.00}{80.00}{721}{119.00}{77.00}{722}
\put(119.00,81.00){\circle{2.00}}
\emline{120.00}{81.00}{723}{123.00}{81.00}{724}
\put(114.00,76.00){\circle{2.00}}
\emline{114.00}{75.00}{725}{114.00}{72.00}{726}
\emline{115.00}{76.00}{727}{118.00}{76.00}{728}
\emline{74.00}{30.00}{729}{74.00}{27.00}{730}
\put(79.00,26.00){\circle{2.00}}
\emline{75.00}{26.00}{731}{78.00}{26.00}{732}
\emline{79.00}{25.00}{733}{79.00}{22.00}{734}
\emline{49.00}{35.00}{735}{49.00}{32.00}{736}
\emline{54.00}{35.00}{737}{54.00}{32.00}{738}
\emline{59.00}{35.00}{739}{59.00}{32.00}{740}
\emline{64.00}{35.00}{741}{64.00}{32.00}{742}
\emline{69.00}{35.00}{743}{69.00}{32.00}{744}
\put(74.00,31.00){\circle{2.00}}
\emline{70.00}{31.00}{745}{73.00}{31.00}{746}
\put(49.00,36.00){\circle{2.00}}
\put(54.00,36.00){\circle{2.00}}
\put(59.00,36.00){\circle{2.00}}
\put(64.00,36.00){\circle{2.00}}
\put(69.00,36.00){\circle{2.00}}
\emline{50.00}{36.00}{747}{53.00}{36.00}{748}
\emline{55.00}{36.00}{749}{58.00}{36.00}{750}
\emline{60.00}{36.00}{751}{63.00}{36.00}{752}
\emline{65.00}{36.00}{753}{68.00}{36.00}{754}
\put(74.00,36.00){\circle{2.00}}
\emline{154.00}{30.00}{755}{154.00}{27.00}{756}
\put(159.00,26.00){\circle{2.00}}
\emline{155.00}{26.00}{757}{158.00}{26.00}{758}
\emline{159.00}{25.00}{759}{159.00}{22.00}{760}
\emline{129.00}{35.00}{761}{129.00}{32.00}{762}
\emline{134.00}{35.00}{763}{134.00}{32.00}{764}
\emline{139.00}{35.00}{765}{139.00}{32.00}{766}
\emline{144.00}{35.00}{767}{144.00}{32.00}{768}
\emline{149.00}{35.00}{769}{149.00}{32.00}{770}
\put(154.00,31.00){\circle{2.00}}
\emline{150.00}{31.00}{771}{153.00}{31.00}{772}
\put(129.00,36.00){\circle{2.00}}
\put(134.00,36.00){\circle{2.00}}
\put(139.00,36.00){\circle{2.00}}
\put(144.00,36.00){\circle{2.00}}
\put(149.00,36.00){\circle{2.00}}
\emline{130.00}{36.00}{773}{133.00}{36.00}{774}
\emline{135.00}{36.00}{775}{138.00}{36.00}{776}
\emline{140.00}{36.00}{777}{143.00}{36.00}{778}
\emline{145.00}{36.00}{779}{148.00}{36.00}{780}
\emline{44.00}{35.00}{781}{44.00}{32.00}{782}
\put(44.00,36.00){\circle{2.00}}
\emline{45.00}{36.00}{783}{48.00}{36.00}{784}
\emline{39.00}{30.00}{785}{39.00}{27.00}{786}
\put(39.00,31.00){\circle{2.00}}
\emline{40.00}{31.00}{787}{43.00}{31.00}{788}
\put(34.00,26.00){\circle{2.00}}
\emline{34.00}{25.00}{789}{34.00}{22.00}{790}
\emline{35.00}{26.00}{791}{38.00}{26.00}{792}
\emline{124.00}{35.00}{793}{124.00}{32.00}{794}
\put(124.00,36.00){\circle{2.00}}
\emline{125.00}{36.00}{795}{128.00}{36.00}{796}
\emline{119.00}{30.00}{797}{119.00}{27.00}{798}
\put(119.00,31.00){\circle{2.00}}
\emline{120.00}{31.00}{799}{123.00}{31.00}{800}
\put(114.00,26.00){\circle{2.00}}
\emline{114.00}{25.00}{801}{114.00}{22.00}{802}
\emline{115.00}{26.00}{803}{118.00}{26.00}{804}
\put(59.00,46.00){\circle{2.00}}
\put(64.00,46.00){\circle{2.00}}
\emline{59.00}{50.00}{805}{59.00}{47.00}{806}
\put(69.00,46.00){\circle{2.00}}
\put(74.00,46.00){\circle{2.00}}
\put(79.00,46.00){\circle{2.00}}
\emline{64.00}{50.00}{807}{64.00}{47.00}{808}
\emline{69.00}{50.00}{809}{69.00}{47.00}{810}
\emline{74.00}{50.00}{811}{74.00}{47.00}{812}
\emline{79.00}{50.00}{813}{79.00}{47.00}{814}
\emline{60.00}{46.00}{815}{63.00}{46.00}{816}
\emline{65.00}{46.00}{817}{68.00}{46.00}{818}
\emline{70.00}{46.00}{819}{73.00}{46.00}{820}
\emline{75.00}{46.00}{821}{78.00}{46.00}{822}
\put(84.00,51.00){\circle{2.00}}
\emline{84.00}{55.00}{823}{84.00}{52.00}{824}
\emline{80.00}{51.00}{825}{83.00}{51.00}{826}
\put(89.00,56.00){\circle{2.00}}
\emline{89.00}{60.00}{827}{89.00}{57.00}{828}
\emline{85.00}{56.00}{829}{88.00}{56.00}{830}
\put(129.00,46.00){\circle{2.00}}
\put(134.00,46.00){\circle{2.00}}
\emline{129.00}{50.00}{831}{129.00}{47.00}{832}
\put(139.00,46.00){\circle{2.00}}
\put(144.00,46.00){\circle{2.00}}
\put(149.00,46.00){\circle{2.00}}
\emline{134.00}{50.00}{833}{134.00}{47.00}{834}
\emline{139.00}{50.00}{835}{139.00}{47.00}{836}
\emline{144.00}{50.00}{837}{144.00}{47.00}{838}
\emline{149.00}{50.00}{839}{149.00}{47.00}{840}
\emline{130.00}{46.00}{841}{133.00}{46.00}{842}
\emline{135.00}{46.00}{843}{138.00}{46.00}{844}
\emline{140.00}{46.00}{845}{143.00}{46.00}{846}
\emline{145.00}{46.00}{847}{148.00}{46.00}{848}
\put(154.00,46.00){\circle{2.00}}
\put(157.00,46.00){\makebox(0,0)[cc]{\small{?}}}
\put(54.00,46.00){\circle{2.00}}
\emline{54.00}{50.00}{849}{54.00}{47.00}{850}
\emline{55.00}{46.00}{851}{58.00}{46.00}{852}
\put(49.00,51.00){\circle{2.00}}
\emline{49.00}{55.00}{853}{49.00}{52.00}{854}
\emline{50.00}{51.00}{855}{53.00}{51.00}{856}
\put(44.00,56.00){\circle{2.00}}
\emline{44.00}{60.00}{857}{44.00}{57.00}{858}
\emline{45.00}{56.00}{859}{48.00}{56.00}{860}
\put(124.00,46.00){\circle{2.00}}
\emline{124.00}{50.00}{861}{124.00}{47.00}{862}
\emline{125.00}{46.00}{863}{128.00}{46.00}{864}
\emline{150.00}{51.00}{865}{153.00}{51.00}{866}
\emline{154.00}{55.00}{867}{154.00}{52.00}{868}
\put(159.00,56.00){\circle{2.00}}
\emline{159.00}{60.00}{869}{159.00}{57.00}{870}
\emline{155.00}{56.00}{871}{158.00}{56.00}{872}
\emline{119.00}{60.00}{873}{119.00}{57.00}{874}
\put(119.00,51.00){\circle{2.00}}
\emline{119.00}{55.00}{875}{119.00}{52.00}{876}
\emline{120.00}{51.00}{877}{123.00}{51.00}{878}
\put(114.00,56.00){\circle{2.00}}
\emline{114.00}{60.00}{879}{114.00}{57.00}{880}
\emline{115.00}{56.00}{881}{118.00}{56.00}{882}
\put(29.00,16.00){\circle{2.00}}
\emline{29.00}{20.00}{883}{29.00}{17.00}{884}
\put(29.00,21.00){\circle{2.00}}
\emline{30.00}{16.00}{885}{33.00}{16.00}{886}
\emline{30.00}{21.00}{887}{33.00}{21.00}{888}
\put(29.00,26.00){\circle{2.00}}
\emline{29.00}{25.00}{889}{29.00}{22.00}{890}
\emline{30.00}{26.00}{891}{33.00}{26.00}{892}
\put(94.00,61.00){\circle{2.00}}
\emline{94.00}{65.00}{893}{94.00}{62.00}{894}
\emline{90.00}{61.00}{895}{93.00}{61.00}{896}
\put(94.00,71.00){\circle{2.00}}
\emline{90.00}{71.00}{897}{93.00}{71.00}{898}
\put(94.00,66.00){\circle{2.00}}
\emline{94.00}{70.00}{899}{94.00}{67.00}{900}
\emline{90.00}{66.00}{901}{93.00}{66.00}{902}
\put(84.00,16.00){\circle{2.00}}
\emline{84.00}{20.00}{903}{84.00}{17.00}{904}
\emline{80.00}{16.00}{905}{83.00}{16.00}{906}
\put(84.00,21.00){\circle{2.00}}
\emline{80.00}{21.00}{907}{83.00}{21.00}{908}
\put(164.00,61.00){\circle{2.00}}
\emline{164.00}{65.00}{909}{164.00}{62.00}{910}
\emline{160.00}{61.00}{911}{163.00}{61.00}{912}
\put(164.00,71.00){\circle{2.00}}
\emline{160.00}{71.00}{913}{163.00}{71.00}{914}
\put(164.00,66.00){\circle{2.00}}
\emline{164.00}{70.00}{915}{164.00}{67.00}{916}
\emline{160.00}{66.00}{917}{163.00}{66.00}{918}
\put(164.00,16.00){\circle{2.00}}
\emline{164.00}{20.00}{919}{164.00}{17.00}{920}
\emline{160.00}{16.00}{921}{163.00}{16.00}{922}
\put(164.00,21.00){\circle{2.00}}
\emline{160.00}{21.00}{923}{163.00}{21.00}{924}
\put(84.00,26.00){\circle{2.00}}
\emline{80.00}{26.00}{925}{83.00}{26.00}{926}
\emline{84.00}{25.00}{927}{84.00}{22.00}{928}
\put(164.00,26.00){\circle{2.00}}
\emline{160.00}{26.00}{929}{163.00}{26.00}{930}
\emline{164.00}{25.00}{931}{164.00}{22.00}{932}
\put(39.00,61.00){\circle{2.00}}
\emline{39.00}{65.00}{933}{39.00}{62.00}{934}
\put(39.00,71.00){\circle{2.00}}
\put(39.00,66.00){\circle{2.00}}
\emline{39.00}{70.00}{935}{39.00}{67.00}{936}
\emline{40.00}{61.00}{937}{43.00}{61.00}{938}
\emline{40.00}{71.00}{939}{43.00}{71.00}{940}
\emline{40.00}{66.00}{941}{43.00}{66.00}{942}
\put(109.00,61.00){\circle{2.00}}
\emline{109.00}{65.00}{943}{109.00}{62.00}{944}
\put(109.00,71.00){\circle{2.00}}
\put(109.00,66.00){\circle{2.00}}
\emline{109.00}{70.00}{945}{109.00}{67.00}{946}
\emline{110.00}{61.00}{947}{113.00}{61.00}{948}
\emline{110.00}{71.00}{949}{113.00}{71.00}{950}
\emline{110.00}{66.00}{951}{113.00}{66.00}{952}
\put(109.00,16.00){\circle{2.00}}
\emline{109.00}{20.00}{953}{109.00}{17.00}{954}
\put(109.00,21.00){\circle{2.00}}
\emline{110.00}{16.00}{955}{113.00}{16.00}{956}
\emline{110.00}{21.00}{957}{113.00}{21.00}{958}
\put(109.00,26.00){\circle{2.00}}
\emline{109.00}{25.00}{959}{109.00}{22.00}{960}
\emline{110.00}{26.00}{961}{113.00}{26.00}{962}
\emline{55.00}{71.00}{963}{78.00}{71.00}{964}
\emline{55.00}{66.00}{965}{78.00}{66.00}{966}
\emline{55.00}{61.00}{967}{78.00}{61.00}{968}
\emline{54.00}{62.00}{969}{54.00}{70.00}{970}
\emline{59.00}{62.00}{971}{59.00}{70.00}{972}
\emline{64.00}{62.00}{973}{64.00}{70.00}{974}
\emline{69.00}{62.00}{975}{69.00}{70.00}{976}
\emline{74.00}{62.00}{977}{74.00}{70.00}{978}
\emline{79.00}{62.00}{979}{79.00}{70.00}{980}
\emline{125.00}{71.00}{981}{148.00}{71.00}{982}
\emline{125.00}{66.00}{983}{148.00}{66.00}{984}
\emline{125.00}{61.00}{985}{148.00}{61.00}{986}
\emline{124.00}{62.00}{987}{124.00}{70.00}{988}
\emline{129.00}{62.00}{989}{129.00}{70.00}{990}
\emline{134.00}{62.00}{991}{134.00}{70.00}{992}
\emline{139.00}{62.00}{993}{139.00}{70.00}{994}
\emline{144.00}{62.00}{995}{144.00}{70.00}{996}
\emline{149.00}{62.00}{997}{149.00}{70.00}{998}
\emline{45.00}{26.00}{999}{68.00}{26.00}{1000}
\emline{45.00}{21.00}{1001}{68.00}{21.00}{1002}
\emline{45.00}{16.00}{1003}{68.00}{16.00}{1004}
\emline{44.00}{17.00}{1005}{44.00}{25.00}{1006}
\emline{49.00}{17.00}{1007}{49.00}{25.00}{1008}
\emline{54.00}{17.00}{1009}{54.00}{25.00}{1010}
\emline{59.00}{17.00}{1011}{59.00}{25.00}{1012}
\emline{64.00}{17.00}{1013}{64.00}{25.00}{1014}
\emline{69.00}{17.00}{1015}{69.00}{25.00}{1016}
\emline{125.00}{26.00}{1017}{148.00}{26.00}{1018}
\emline{125.00}{21.00}{1019}{148.00}{21.00}{1020}
\emline{125.00}{16.00}{1021}{148.00}{16.00}{1022}
\emline{124.00}{17.00}{1023}{124.00}{25.00}{1024}
\emline{129.00}{17.00}{1025}{129.00}{25.00}{1026}
\emline{144.00}{17.00}{1027}{144.00}{25.00}{1028}
\emline{149.00}{17.00}{1029}{149.00}{25.00}{1030}
\emline{134.00}{17.00}{1031}{134.00}{25.00}{1032}
\emline{139.00}{17.00}{1033}{139.00}{25.00}{1034}
\end{picture}
\end{figure}

\section{Proof of Theorem \protect\ref{T2}}

\begin{proof}
We will construct the desired $\mathrm{PDDS}$ by applying Corollary \ref{C}.
Set $H=Q_{3}$. We place the graph $D=(V,E)$ that is isomorphic to $H^{\ast }$
in such a way that $V$ comprises the vertices $O$, $e_{1}$, $e_{2}$, $e_{3}$%
, $e_{1}+e_{2}$, $e_{1}+e_{3}$, $e_{2}+e_{3}$ and $e_{1}+e_{2}+e_{3}$ of $%
Q_{3}$ and their $24$ neighbors. Thus, $|V|=32,$ and $D$ contains the
vertices $O$ and $e_{i}$, for $i=1,2,3$, as required by Corollary \ref{C}.
We choose $G=\mathbb{Z}_{2}\oplus \mathbb{Z}_{4}\oplus \mathbb{Z}_{4}$. The
elements $g_{i}$ of $G$ that are assigned to the vertices $e_{i}$ are: $%
g_{1}=1,3,3$, $g_{2}=0,1,0$ and $g_{3}=0,0,1$. To finish the proof, we need
to show that the restriction of the mapping $\Phi ((a_{1},a_{2},a_{3}))=\Phi
(e_{1})^{a_{1}}\circ \Phi (e_{2})^{a_{2}}\circ \Phi
(e_{3})^{a_{3}}=a_{1}g_{1}+a_{2}g_{2}+a_{3}g_{3}$ to the set $V$ is a
bijection. For the reader's convenience we provide all values of $\Phi $ on $%
V$ in a table below. It suffices to note that all these values are distinct.
The vertices in $V$ are given in the left-hand side of the table, the
corresponding values of $\Phi $ in the right-hand side.

$$
\begin{array}{||l|l|l|l||l|l|l|l||}
\hline
& _{-e_3} & _{e_1-e_3} &  &  & _{0,0,3} & _{1,3,2} &  \\
& ^{e_2-e_3} & ^{e_1+e_2-e_3} &  &  & ^{0,1,3} & ^{1,0,2} &  \\ \hline
_{-e_1} & _{O}^{-e_2} & _{e_1}^{e_1-e_2} & _{2e_1} & _{1,1,1} &
_{0,0,0}^{0,3,0} & _{1,3,3}^{1,2,3} & _{0,2,2} \\
^{e_2-e_1} & _{2e_2}^{e_2} & _{e_1+2e_2}^{e_1+e_2} & ^{2e_1+e_2} & ^{1,2,1}
& _{0,2,0}^{0,1,0} & _{1,1,3}^{1,0,3} & ^{0,3,2} \\ \hline
_{e_3-e_1} & _{e_3}^{e_3-e_2} & _{e_1+e_3}^{e_1-e_2+e_3} & _{2e_1+e_3} &
_{1,1,2} & _{0,0,1}^{0,3,1} & _{1,3,0}^{1,2,0} & _{0,2,3} \\
^{e_2+e_3-e_1} & _{2e_2+e_3}^{e_2+e_3} & _{e_1+2e_2+e_3}^{e_1+e_2+e_3} &
^{2e_1+e_2+e_3} & ^{1,2,2} & _{0,2,1}^{0,1,1} & _{1,1,0}^{1,0,0} & ^{0,3,3}
\\ \hline
& _{2e_3} & _{e_1+2e_3} &  &  & _{0,0,2} & _{1,3,1} &  \\
& ^{e_2+2e_3} & ^{e_1+e_2+2e_3} &  &  & ^{0,1,2} & ^{1,1,0} &  \\ \hline
\end{array}%
$$
\end{proof}

\section{A periodic 1-$\mathrm{PDDS}[P_{2}]$ that is not lattice-like}

\noindent Here we provide a periodic $1$-$\mathrm{PDDS}[P_{2}]$ $R$ that is
not lattice-like. To see this it will suffice to notice that some components
of $R$ are paths $P_{2}$ "parallel to $x$-axis", some "parallel to $y$%
-axis". A typical part of $R$ consisting of four copies of $P_{2}$ and their
neighbors is provided in the figure below:

\begin{center}
\begin{figure}[htp]
\unitlength=0.50mm \special{em:linewidth 0.4pt} \linethickness{0.4pt}
\begin{picture}(137.00,27.00)
\put(116.00,1.00){\circle{2.00}}
\put(121.00,1.00){\circle{2.00}}
\put(106.00,6.00){\circle{2.00}}
\put(111.00,6.00){\circle{2.00}}
\put(126.00,6.00){\circle{2.00}}
\put(131.00,6.00){\circle{2.00}}
\put(101.00,11.00){\circle{2.00}}
\put(111.00,11.00){\circle{2.00}}
\put(116.00,11.00){\circle{2.00}}
\put(121.00,11.00){\circle{2.00}}
\put(126.00,11.00){\circle{2.00}}
\put(136.00,11.00){\circle{2.00}}
\put(101.00,16.00){\circle{2.00}}
\put(111.00,16.00){\circle{2.00}}
\put(116.00,16.00){\circle{2.00}}
\put(121.00,16.00){\circle{2.00}}
\put(126.00,16.00){\circle{2.00}}
\put(136.00,16.00){\circle{2.00}}
\put(106.00,21.00){\circle{2.00}}
\put(111.00,21.00){\circle{2.00}}
\put(116.00,21.00){\circle{2.00}}
\put(121.00,21.00){\circle{2.00}}
\put(126.00,21.00){\circle{2.00}}
\put(131.00,21.00){\circle{2.00}}
\put(116.00,26.00){\circle{2.00}}
\put(121.00,26.00){\circle{2.00}}
\put(116.00,6.00){\circle*{2.50}}
\put(121.00,6.00){\circle*{2.50}}
\emline{117.00}{6.00}{1}{120.00}{6.00}{2}
\emline{116.00}{10.00}{3}{116.00}{7.00}{4}
\emline{115.00}{6.00}{5}{112.00}{6.00}{6}
\emline{116.00}{5.00}{7}{116.00}{2.00}{8}
\emline{121.00}{7.00}{9}{121.00}{10.00}{10}
\emline{122.00}{6.00}{11}{125.00}{6.00}{12}
\emline{121.00}{5.00}{13}{121.00}{2.00}{14}
\put(116.00,21.00){\circle*{2.50}}
\put(121.00,21.00){\circle*{2.50}}
\emline{117.00}{21.00}{15}{120.00}{21.00}{16}
\emline{116.00}{25.00}{17}{116.00}{22.00}{18}
\emline{115.00}{21.00}{19}{112.00}{21.00}{20}
\emline{116.00}{20.00}{21}{116.00}{17.00}{22}
\emline{121.00}{22.00}{23}{121.00}{25.00}{24}
\emline{122.00}{21.00}{25}{125.00}{21.00}{26}
\emline{121.00}{20.00}{27}{121.00}{17.00}{28}
\put(106.00,16.00){\circle*{2.50}}
\put(106.00,11.00){\circle*{2.50}}
\put(131.00,16.00){\circle*{2.50}}
\put(131.00,11.00){\circle*{2.50}}
\emline{132.00}{11.00}{29}{135.00}{11.00}{30}
\emline{131.00}{15.00}{31}{131.00}{12.00}{32}
\emline{132.00}{16.00}{33}{135.00}{16.00}{34}
\emline{131.00}{10.00}{35}{131.00}{7.00}{36}
\emline{127.00}{11.00}{37}{130.00}{11.00}{38}
\emline{130.00}{16.00}{39}{127.00}{16.00}{40}
\emline{131.00}{20.00}{41}{131.00}{17.00}{42}
\emline{107.00}{11.00}{43}{110.00}{11.00}{44}
\emline{106.00}{15.00}{45}{106.00}{12.00}{46}
\emline{107.00}{16.00}{47}{110.00}{16.00}{48}
\emline{106.00}{10.00}{49}{106.00}{7.00}{50}
\emline{102.00}{11.00}{51}{105.00}{11.00}{52}
\emline{105.00}{16.00}{53}{102.00}{16.00}{54}
\emline{106.00}{20.00}{55}{106.00}{17.00}{56}
\end{picture}
\end{figure}
\end{center}

\noindent Despite the fact that $R$ is not lattice-like we will show how it
is possible to construct it by means of a slight modification of Corollary %
\ref{C}.\bigskip

\noindent We take $H^*$ to be a graph induced by the 32 vertices in the
figure above. To obtain the graph $D=(V,E)$ we place $H^*$ so that the four
copies of $P_2$ occupy vertices $(0,1)$ and $(1,1);$ $($ $0,-2)$ and $%
(1,-2); $ $(-2,-1)$ and $(-2,0);$ and finally $(3,-1)$ and $(3,0)$
respectively. We choose as $G$ the group $\mathbb{Z}_4\oplus \mathbb{Z}_8.$
The elements of $G $ assigned to $e_1$ and $e_2$ are $0,1$ and $1,1$
respectively. The restriction of the homomorphism $\Phi$ to $V$ is provided
below in the matrix form. It is easy to verify from the matrix that $\Phi$
is a bijection on $V.$

%\newpage

\begin{center}
\[
\begin{array}{cccccccc}
&  &  & 2,6 & 2,7 &  &  &  \\
& 3,5 & 3,6 & \mathbf{3,7} & \mathbf{3,0} & 3,1 & 3,2 &  \\
0,5 & \mathbf{0,6} & 0,7 & 0,0 & 0,1 & 0,2 & \mathbf{0,3} & 0,4 \\
1,6 & \mathbf{1,7} & 1,0 & 1,1 & 1,2 & 1,3 & \mathbf{1,4} & 1,5 \\
& 2,0 & 2,1 & \mathbf{2,2} & \mathbf{2,3} & 2,4 & 2,5 &  \\
&  &  & 3,3 & 3,4 &  &  &
\end{array}
\]
\end{center}

\noindent Thus Corollary \ref{C} provides a decomposition of $\mathbb{Z}^{2}$
into parts of order $32,$ each of them isomorphic to $H^{\ast }.$ Further,
as $H^{\ast }$ can be decomposed into four copies of $P_{2}$ and its
neighbors, we have constructed a $1$-$\mathrm{PDDS}[P_{2}]$ $R$ that is not
lattice-like. However, it is straightforward that $R$ is periodic. Therefore
we have proved:

\begin{theorem}
There exists a periodic non-lattice-like $1$-$\mathrm{PDDS}[P_{2}]$ in $%
\Lambda _{2}$.
\end{theorem}

\textbf{Acknowledgement.} We thank Ana Breda from the University of Aveiro
for her comments that helped to improve presentation of this paper. We also
thank G. Mazzuoccolo from the University of Modena, who provided an example
of a $2$-$\mathrm{PDDS}$ whose components are all isomorphic to $%
P_{2}\square P_{2}$.

\end{document}